\documentclass[copyright,creativecommons]{eptcs}

\providecommand{\event}{ICE 2022} 
\usepackage{breakurl}             
\usepackage{underscore}           

\usepackage{verbatim}
\usepackage{amsfonts}
\usepackage{amsmath}
\usepackage{amssymb}
\usepackage{amsthm}

\usepackage{xspace}
\usepackage{float}
\usepackage{tikz}
\usepackage{pgf}
\usepackage{bussproofs}
\usepackage{framed}
\usepackage{listings}
\usepackage[nameinlink]{cleveref}
\usepackage{subcaption}
\usepackage{xfrac}
\usepackage{mathtools}
\usepackage{url}
\usepackage[numbers, sort, comma, square]{natbib}
\usepackage{enumerate}
\usepackage{cases}
\usepackage{xkcdcolors} 
\usepackage{thmtools}
\usepackage{bold-extra} 
\usepackage{paralist} 
\usepackage{booktabs} 
\usepackage{multirow} 
\usepackage{eurosym} 
\usepackage{fixmath} 

\hypersetup{
    colorlinks,
    linkcolor={red!50!black},
    citecolor={blue!50!black},
    urlcolor={blue!80!black}
}

\declaretheoremstyle[ 
  spaceabove=12pt, 
  spacebelow=6pt, 
  bodyfont=\itshape
]{thmboxstyle}
\declaretheoremstyle[
  spaceabove=12pt, 
  spacebelow=6pt, 
  bodyfont=\normalfont,
  notefont=\normalfont\itshape,
  qed=$\blacksquare$,
  spaceabove=4mm, 
  spacebelow=4mm
]{definitionstyle}
\declaretheoremstyle[
  bodyfont=\itshape,
  qed=$\blacksquare$,
  spaceabove=4mm, 
  spacebelow=4mm
]{remarkstyle}
\declaretheoremstyle[]{nostyle}

\declaretheorem[name=Example,style=definitionstyle,numberwithin=section,Refname={Example,Examples}]{example}
\declaretheorem[name=Definition,style=definitionstyle,numberwithin=section,Refname={Definition,Definitions}]{definition}

\usepackage{alphalph}
\usepackage{etoolbox}

\AtBeginDocument{%
  \AtBeginEnvironment{subequations}{%
  }
}

\lstdefinelanguage{elixir}{
  morekeywords={case,catch,def,do,else,false,%
  use,alias,receive,defmacro,defp,%
  for,if,import,defmodule,defprotocol,%
  nil,defmacrop,defoverridable,defimpl,%
  super,fn,raise,true,try,end,with,unless, %
  @dual, @spec, @session, send,%
  and, when},
  otherkeywords={<-,->},
  emph={self},
  commentstyle=\itshape\color{gray},
  sensitive=true,
  morecomment=[l]{\#},
  morecomment=[n]{/*}{*/},
  morestring=[b]",
  morestring=[b]',
  morestring=[b]"""
}

\newcommand{\lstprompt}{>>>}
\newcommand\numberwithprompt[1]{\footnotesize\ttfamily\selectfont \lstprompt}

\lstdefinelanguage{output}{
  numbers=none
}

\crefname{figure}{figure}{figures}
\crefname{section}{section}{sections}
\crefname{chapter}{chapter}{chapters}
\crefname{table}{table}{tables}
\crefname{listing}{listing}{listings}
\crefname{prp}{property}{properties}
\crefname{frm}{formula}{formulae}
\crefname{case}{case}{cases}
\crefname{item}{item}{items}
\crefname{objective}{objective}{objectives}
\crefname{expression}{expression}{expressions}
\crefname{inlineitem}{}{}

\creflabelformat{item}{#2$#1$#3}
\creflabelformat{inlineitem}{\itshape #2(#1)#3}
\creflabelformat{objective}{#2O#1#3}
\creflabelformat{expression}{#2(#1)#3}

\crefrangelabelformat{equation}%
{(#3#1#4--#5\crefstripprefix{#1}{#2}#6)}
\crefrangelabelformat{subequation}%
{(#3#1#4--#5\crefstripprefix{#1}{#2}#6)}
\crefrangelabelformat{expression}%
{(#3#1#4--#5\crefstripprefix{#1}{#2}#6)}

\setdefaultenum{\itshape (i)}{}{}{}

\usetikzlibrary {
  matrix,
  shapes,
  arrows,
  shadows,
  calc,
  chains,
  arrows.meta,
  patterns,
  fit,
  bending,
  shadows,
  shapes.geometric,
  tikzmark,
  decorations.pathmorphing,
  decorations.text,
  decorations.pathreplacing,
  decorations.markings
}

\tikzset{
  label/.style={
    font=\scriptsize\itshape,
    inner sep=0.4em
  },
  state/.style={
    circle,
    text width=0.8em,
    inner sep=0.1em,
    text depth=0.08em,
    draw,
    font=\scriptsize
  },
  participant/.style={
    draw=black,
    rounded corners,
    semithick,
    font=\sffamily\normalsize,
    text height=0.3cm,
    text centered,
    anchor=base
  },
  stage/.style={
    draw=black,
    font=\sffamily\small,
    text height=0.2cm,
    minimum height=0.55cm,
    text centered
  },
  spawn/.style={
    draw=black!50!green,
    thick,
    font=\sffamily\small,
    text height=0.2cm,
    minimum height=0.55cm,
    text centered
  },
  function/.style={
    draw=black!50!red,
    font=\sffamily\small,
    text height=0.2cm,
    minimum height=0.55cm,
    text centered
  },
  point/.style={
    minimum size=0pt, 
    inner sep=0pt
  },
  clientprocess/.style={
    draw=black,
    fill=white,
    semithick,
    minimum width=0.5cm,
    minimum height=0.5cm,
    font=\sffamily\small,
    text height=0.2cm,
    inner sep=0.4cm,
    line width=0.5mm
  },
  negated/.style={
        decoration={markings,
            mark= at position 0.5 with {
                \node[transform shape] (tempnode) {$\backslash$};
            }
        },
        postaction={decorate}
  },
  ashadow/.style={
    opacity=.3, 
    shadow xshift=1.2mm, 
    shadow yshift=-1.2mm 
  }
}

\newcommand{\elixircode}[1]{\lstinline[language=elixir, breaklines=false]$#1$}

\newcommand{\mycomment}[1]{}
\newcommand*{\ElixirST}{\textsf{ElixirST}\xspace}
\newcommand*{\Elixir}{Elixir\xspace}
\newcommand*{\Erlang}{Erlang\xspace}

\newcommand*{\pid}{\emph{pid}\xspace}
\newcommand*{\pids}{\emph{pids}\xspace}
\newcommand*{\eg}{\textit{e.g.}\xspace}

\newcommand*{\ie}{\textit{i.e.,}\xspace}

\newcommand*{\wrt}{w.r.t.\xspace}
\newcommand*{\etal}{\textit{et~al.}\xspace}
\newcommand{\defsymbol}{\ensuremath{\stackrel{\text{\tiny def}}{=}}}

\newcommand{\highlight}[2]{\colorbox{#1!17}{$\displaystyle #2$}}

\renewcommand{\highlight}[2]{\colorbox{#1!17}{#2}}

\newcommand{\blueColor}[1]{\textcolor{blue}{#1}}
\newcommand{\greenColor}[1]{\textcolor{black!50!green}{#1}}
\definecolor{darkgray}{rgb}{0.2, 0.2, 0.2}

\colorlet{RED}{red}

\newcommand{\fullref}[1]  {\nameref{#1} \nameCref{#1}\xspace}

\newcommand{\stsessionbranch}[2]  {\ensuremath{\& \big\{ #1 \big\}_{#2}}}
\newcommand{\stsessionchoice}[2]  {\ensuremath{\oplus \big\{ #1 \big\}_{#2}}}
\newcommand{\stsessionrec}[2]     {\ensuremath{\textsf{rec} \; \stsessionrecvar{#1} \, . \, #2}}
\newcommand{\stsessionrecvar}[1]	{\ensuremath{\textsf{#1}}}
\newcommand{\stend}               {\ensuremath{\textsf{end}}}

\newcommand{\texttype}[1]	        {\ensuremath{\textsf{#1}}}
\newcommand{\textlabel}[1]	      {\ensuremath{\texttt{#1}}}
\newcommand{\basicvalues}[1]	      {\ensuremath{\textit{#1}}}

\newcommand{\dualpid}[0]		    {\ensuremath{y}}
\newcommand{\identifier}[0]		    {\ensuremath{w}}
\newcommand{\pidvalue}[0]           {\ensuremath{\iota}}

\newcommand{\prooflabel}[1]			{\ensuremath{\left[\textsc{#1}\right]}}
\newcommand{\prooflabell}[2]		{\ensuremath{\left[\textsc{#1}_\textsc{#2}\right]}}

\newcommand{\stmodule}[3]			{\ensuremath{\texttt{defmodule} \; #1 \; \texttt{do} \; #2 \; #3 \; \texttt{end}}}
\newcommand{\stdef}[3]				{\ensuremath{\texttt{def} \; #1 \! \left( #2 \right) \texttt{do} \; #3 \; \texttt{end}}}

\newcommand{\stdefp}[3]				{\ensuremath{\texttt{defp} \; #1 \! \left( #2 \right) \texttt{do} \; #3 \; \texttt{end}}}

\newcommand{\stsessionannotation}[1]{\ensuremath{\texttt{@session} \; ``#1"}}

\newcommand{\stdualannotation}[1]	{\ensuremath{\texttt{@dual} \; ``#1"}}

\newcommand{\stspec}[3]				{\ensuremath{\texttt{@spec} \; #1 \left( #2 \right) \; :: \; #3}}

\newcommand{\stspecbig}[3]				{\ensuremath{\texttt{@spec} \; #1 \big( #2 \big) \; :: \; #3}}

\newcommand{\stexpsequence}[3]		{\ensuremath{ #1 = #2 ; \; #3}}
\newcommand{\stexpsend}[2]			{\ensuremath{ \texttt{send} \left( #1, #2 \right) }}
\newcommand{\stexpreceive}[3]		{\ensuremath{ \texttt{receive do} \; { \left( #1 \rightarrow #2 \right)}_{#3} \texttt{end} }}
\newcommand{\stexpreceiveleft}		{\ensuremath{ \texttt{receive do} \; }}
\newcommand{\stexpreceiveright}[3]		{\ensuremath{ { \left( #1 \rightarrow #2 \right)}_{#3} \texttt{end} }}
\newcommand{\stexpcase}[4]			{\ensuremath{ \texttt{case} \; #1 \; \texttt{do} \; { \left( #2 \rightarrow #3 \right)}_{#4} \texttt{end} }}
\newcommand{\stexpfunction}[2]		{\ensuremath{ #1 \left( #2 \right) }}
\newcommand{\stexpvalue}[1]			{\ensuremath{ #1 }}
\newcommand{\stexpvariable}[1]		{\ensuremath{ #1 }}
\newcommand{\stexpbinaryop}[3]		{\ensuremath{ #1 \; #2 \; #3}}
\newcommand{\stexpnot}[1]			{\ensuremath{ \texttt{not} \; #1}}
\newcommand{\stexpandonly}			{\ensuremath{\texttt{and}}}

\newcommand{\stexporonly}			{\ensuremath{\texttt{or}}}
\newcommand{\tupletype}[1]			{\ensuremath{ \left\{ #1 \right\} }}
\newcommand{\listtype}[1]			{\ensuremath{ \left[ \, #1 \right] \, }}
\newcommand{\listtypee}[2]			{\ensuremath{ \left[ \, #1 \; | \; #2 \, \right] }}
\newcommand{\dotsspace}				{\ensuremath{ , \; \dots , \; }}

\newcommand{\atom}[1]			    {\texttt{:}\hspace*{-1.5pt}\texttt{#1}}
\newcommand{\annotation}[1]			{\elixircode{@}\hspace*{-1.2pt}\elixircode{#1}}

\newcommand{\typing}[7]				{\ensuremath{#1 \cdot #2 \vdash^{#3} #4 \rhd #5 : #6 \lhd #7}}

\newcommand{\typingsigma}[8]		{\ensuremath{#1 \cdot #2 \vdash_{#3}^{#4} #5 \rhd #6 : #7 \lhd #8}}
\newcommand{\typingmodule}[1]		{\ensuremath{\vdash #1}}
\newcommand{\typingexpression}[3]	{\ensuremath{#1 \vdash_\text{exp} #2 : #3}}

\newcommand{\typingpattern}[4]		{\ensuremath{\vdash_\text{pat}^{#1} #2 : #3 \; \rhd \; #4 }}

\newcommand{\fun}[2]                {\ensuremath{#1\hspace{-0.5mm}/\hspace{-0.5mm}#2}}

\newcommand{\expressionsem}[2]	    {\ensuremath{#1 \rightarrow #2}}
\newcommand{\termsem}[3]	        {\ensuremath{#2 \xrightarrow{#1} #3}}
\newcommand{\termsemm}[4]	        {\ensuremath{#3 \xrightarrow[#2]{#1} #4}}

\newcommand{\substitution}[2]	    {\ensuremath{ \left[\sfrac{#1}{#2}\right]}}
\newcommand{\substitutionsingle}[1]	{\ensuremath{ \left[#1\right]}}
\newcommand{\substitutionvx}	    {\ensuremath{ \substitution{v}{x}}}

\newcommand{\actionsend}[2]	        {\ensuremath{#1 ! #2}}
\newcommand{\actionreceive}[1]	    {\ensuremath{? #1}}

\newcommand{\function}[2]	                {\functionname{#1}(#2)}
\newcommand{\functionname}[1]	            {\textbf{#1}\xspace}

\newcommand{\fv}[1]	                {\function{fv}{#1}}

\newcommand{\bv}[1]	                {\function{bv}{#1}}

\newcommand{\vars}[1]	                {\function{vars}{#1}}

\newcommand{\after}[1]	                {\function{after}{#1}}
\newcommand{\aftername}	            {\functionname{after}}

\newcommand{\match}[1]	                {\function{match}{#1}}
\newcommand{\matchname}	            {\functionname{match}}

\newcommand{\sessions}[1]	                {\function{sessions}{#1}}

\newcommand{\details}[1]	                {\function{details}{#1}}
\newcommand{\detailsname}	            {\functionname{details}}

\newcommand{\functions}[1]	                {\function{functions}{#1}}

\newcommand{\typeof}[1]	                {\function{type}{#1}}

\newcommand{\dom}[1]	                {\function{dom}{#1}}

\newtoggle{shortversion}
\toggletrue{shortversion}

\title{
  Session Fidelity for \ElixirST:\\ 
  A Session-Based Type System for \Elixir Modules
  }
\author{
  \hspace{4ex}Gerard Tabone \hspace{23mm} Adrian Francalanza
\institute{
  Computer Science Department\\
  University of Malta\\
  Msida, Malta
  }
\email{
  \hspace{1.7ex}
  gerard.tabone.17@um.edu.mt 
  \hspace{7mm}
  adrian.francalanza@um.edu.mt
  }
}
\def\titlerunning{Session Fidelity for \ElixirST}

\def\authorrunning{G. Tabone \& A. Francalanza}

\begin{document}
\maketitle

\begin{abstract}
This paper builds on prior work investigating the adaptation of session types to provide behavioural information about \Elixir modules. 
A type system called \ElixirST has been constructed to statically determine whether functions in an \Elixir module observe their endpoint specifications, expressed as session types;
a corresponding tool automating this typechecking has also been constructed.
In this paper we formally validate this type system.
An LTS-based operational semantics for the language fragment supported by the type system is developed, modelling its runtime behaviour when invoked by the module client.
This operational semantics is then used to prove session fidelity for \ElixirST.
\end{abstract}

\section{Introduction}
\label{sec:intro}
In order to better utilise recent advances in microprocessor design and architecture distribution, modern programming languages offer a variety of abstractions for the construction of concurrent programs.
In the case of message-passing programs, concurrency manifests itself as spawned computation that exhibits \emph{communication as a side-effect}, potentially influencing the execution of other (concurrent) computation. 
Such side-effects inevitably increase the complexity of the programs produced and lead to new sources of errors.
As a consequence, program correctness becomes harder to verify and language support for detecting errors at the development stage can substantially decrease the number of concurrency errors.

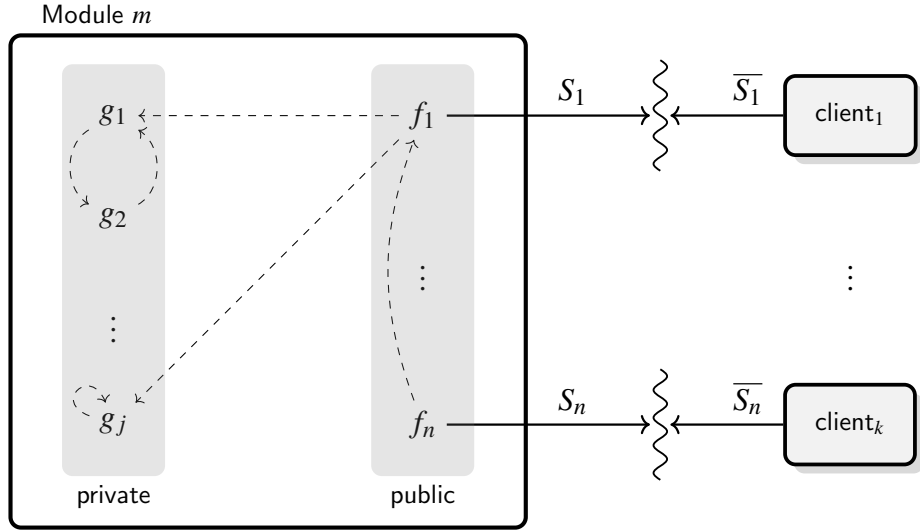
\begin{figure}[t]
	\newif\ifstartcompletesineup
	\newif\ifendcompletesineup
	\pgfkeys{
		/pgf/decoration/.cd,
		start up/.is if=startcompletesineup,
		start up=true,
		start up/.default=true,
		start down/.style={/pgf/decoration/start up=false},
		end up/.is if=endcompletesineup,
		end up=true,
		end up/.default=true,
		end down/.style={/pgf/decoration/end up=false}
	}
	\pgfdeclaredecoration{complete sines}{initial}
	{
		\state{initial}[
			width=+0pt,
			next state=upsine,
			persistent precomputation={
				\ifstartcompletesineup
					\pgfkeys{/pgf/decoration automaton/next state=upsine}
					\ifendcompletesineup
						\pgfmathsetmacro\matchinglength{
							0.5*\pgfdecoratedinputsegmentlength / (ceil(0.5* \pgfdecoratedinputsegmentlength / \pgfdecorationsegmentlength) )
						}
					\else
						\pgfmathsetmacro\matchinglength{
							0.5 * \pgfdecoratedinputsegmentlength / (ceil(0.5 * \pgfdecoratedinputsegmentlength / \pgfdecorationsegmentlength ) - 0.499)
						}
					\fi
				\else
					\pgfkeys{/pgf/decoration automaton/next state=downsine}
					\ifendcompletesineup
						\pgfmathsetmacro\matchinglength{
							0.5* \pgfdecoratedinputsegmentlength / (ceil(0.5 * \pgfdecoratedinputsegmentlength / \pgfdecorationsegmentlength ) - 0.4999)
						}
					\else
						\pgfmathsetmacro\matchinglength{
							0.5 * \pgfdecoratedinputsegmentlength / (ceil(0.5 * \pgfdecoratedinputsegmentlength / \pgfdecorationsegmentlength ) )
						}
					\fi
				\fi
				\setlength{\pgfdecorationsegmentlength}{\matchinglength pt}
			}] {}
		\state{downsine}[width=\pgfdecorationsegmentlength,next state=upsine]{
			\pgfpathsine{\pgfpoint{0.5\pgfdecorationsegmentlength}{0.5\pgfdecorationsegmentamplitude}}
			\pgfpathcosine{\pgfpoint{0.5\pgfdecorationsegmentlength}{-0.5\pgfdecorationsegmentamplitude}}
		}
		\state{upsine}[width=\pgfdecorationsegmentlength,next state=downsine]{
			\pgfpathsine{\pgfpoint{0.5\pgfdecorationsegmentlength}{-0.5\pgfdecorationsegmentamplitude}}
			\pgfpathcosine{\pgfpoint{0.5\pgfdecorationsegmentlength}{0.5\pgfdecorationsegmentamplitude}}
	}
		\state{final}{}
	}
	\centering
	\begin{tikzpicture}[scale=1.37]

		\node (private-f-1) at (1,4) {\large $g_1$};
		\node (private-f-2) at (1,3) {\large $g_2$};
		\node (private-f-m-1) at (1,2) {\large $\vdots$};
		\node (private-f-m) at (1,1) {\large $g_j$};
		%
		
		\node (public-f-1) at (4,4) {\large $f_1$};
		\node (public-f-n-1) at (4,2.5) {\large $\vdots$};
		\node (public-f-n) at (4,1) {\large $f_n$};

		\node[font=\sffamily\small,anchor=south west] (module) at (0.2,4.8cm) {Module $m$};

		\node (f-1-client) [right=2.7 of public-f-1]{};
		\node (f-n-client) [right=2.7 of public-f-n]{};	
		\node[clientprocess, rounded corners, fill=gray!10, drop shadow=ashadow] (client-1) [right=1.5cm of f-1-client] {$\textsf{client}_1$};
		\node[clientprocess, rounded corners, fill=gray!10, drop shadow=ashadow] (client-n) [right=1.5cm of f-n-client] {$\textsf{client}_k$};
		\node (client-dots) at ($(client-1)!0.5!(client-n)$) {\large $\vdots$};	
		
		\node (S-1-temp) [above=0.08cm of f-1-client] {};
		\node (S-n-temp) [above=0.08cm of f-n-client] {};
		\node (S-1) [left=0.7cm of S-1-temp] {\large $S_1\vphantom{\overline{S_1}}$};
		\node (S-1-dual) [right=0.7cm of S-1-temp] {\large $\overline{S_1}$};
		\node (S-n) [left=0.7cm of S-n-temp] {\large $S_n\vphantom{\overline{S_1}}$};
		\node (S-n-dual) [right=0.7cm of S-n-temp] {\large $\overline{S_n}$};

		\node (f-1-mid-above) [above=6mm of f-1-client]{};
		\node (f-1-mid-below) [below=6mm of f-1-client]{};
		\node (f-n-mid-above) [above=6mm of f-n-client]{};
		\node (f-n-mid-below) [below=6mm of f-n-client]{};
		\begin{scope}[every path/.style={
				decoration={
					complete sines,
					segment length=2.2mm,
					amplitude=1.8mm,
					mirror,
					start up,
					end up
				},
				decorate,
				line width=0.25mm
			}]
			\draw (f-1-mid-above) node [anchor=east] {} -- (f-1-mid-below);
			\draw (f-n-mid-above) node [anchor=east] {} -- (f-n-mid-below);
		\end{scope}

		\draw [line width=0.3mm, ->] (public-f-1) to (f-1-client);
		\draw [line width=0.3mm, ->] (client-1) to (f-1-client);
		\draw [line width=0.3mm, ->] (public-f-n) to (f-n-client);
		\draw [line width=0.3mm, ->] (client-n) to (f-n-client);

		\draw [dashed, ->](private-f-1) to[bend right=60] (private-f-2);
		\draw [dashed, ->] (private-f-2) to[bend right=60] (private-f-1);
		\draw [dashed, ->] (public-f-1) to (private-f-1);
		\draw [dashed, ->] (public-f-1) to (private-f-m);
    \draw [dashed, ->] (public-f-n) to[bend left=20] (public-f-1);
		\draw [dashed, ->] (private-f-m) to [looseness=5,in=110, out=160] (private-f-m);

		
		\draw[line width=0.55mm, rounded corners] (0,0) rectangle (5cm,4.78cm);

		\fill[opacity=.2, gray, rounded corners, thin] (0.5cm,0.5cm) rectangle (1.5cm,4.5cm);
		\fill[opacity=.2, gray, rounded corners, thin] (3.5cm,0.5cm) rectangle (4.5cm,4.5cm);
		
		\node[font=\sffamily\small] at (1cm,0.3cm) {private};
		\node[font=\sffamily\small] at (4cm,0.3cm) {public};

	\end{tikzpicture}

	\caption{An \Elixir module consisting of public and private functions, interacting with client processes}
	\label{fig:roadmap}
\end{figure}

\Elixir~\cite{thomas2018programming}, based on the actor model~\cite{DBLP:conf/ijcai/HewittBS73,DBLP:books/daglib/Agha90}, is one such example of a modern programming language for concurrency. 
As depicted in  \Cref{fig:roadmap}, \Elixir programs are structured as a collection of \emph{modules} that contain \emph{functions}, the basic unit of code decomposition in the language.
A module only exposes a subset of these functions to external invocations by defining them as \emph{public}; these functions act as the only entry points to the functionality encapsulated by a module.
Internally, the bodies of these public functions may then invoke other functions, which can either be the \emph{public} ones already exposed or the \emph{private} functions that can only be invoked from within the same module.
For instance, \Cref{fig:roadmap} depicts a module $m$ which contains several public functions (\ie $f_1, \ldots, f_n$) and private functions (\ie $g_1,\ldots,g_j$).
For example, the public function $f_1$ delegates part of its computation by calling the private functions $g_1$ and $g_j$, whereas the body of the public function $f_n$ invokes the other public function $f_1$ when executed.
Internally, the body of the private function $g_1$ calls the other private function $g_2$ whereas the private function $g_j$ can recursively call itself. 


A prevalent \Elixir design pattern is that of a server listening for client requests. 
For each request, the server spawns a (public) function to execute independently and act as a dedicated client handler: after the respective process IDs of the client and the spawned handler are made known to each other, a session of interaction commences between the two concurrent entities (via message-passing).  
For instance, in \Cref{fig:roadmap}, a handler process running public function $f_1$ is assigned to the session with client $\textsf{client}_1$ whereas the request from $\textsf{client}_k$ is assigned a dedicated handler running function $f_n$.
Although traditional interface elements such as function parameters (used to instantiate the executing function body with values such as the client process ID) and the function return value (reporting the eventual outcome of handled request) are important, the messages exchanged between the two concurrent parties within a session are equally important for software correctness.
More specifically, communication incompatibilities between the interacting parties could lead to various runtime errors.
For example, if in a session a message is sent with an unexpected payload, it could cause the receiver's subsequent computation depending on it to crash (\eg multiplying by a string when a number should have been received instead). 
Also, if messages are exchanged in an incorrect order, they may cause \emph{deadlocks} (\eg two processes waiting forever for one another to send messages of a particular kind when a message of a different kind has been sent instead).
%

In many cases, the expected protocol of interactions within a session can be statically determined from the respective endpoint implementations, namely the function bodies;
for simplicity, our discussion assumes that endpoint interaction protocols are dual, \eg $S_1$ and $\overline{S_1}$ in \Cref{fig:roadmap}.
Although \Elixir provides mechanisms for specifying (and checking) the parameters and return values of a function within a module, it does \emph{not} provide any support for describing (and verifying) the interaction protocol of a function in terms of its communication side-effects.
To this end, in earlier work~\cite{DBLP:conf/agere/TaboneF21} we devised the tool\footnote{\ElixirST is available on GitHub: \url{https://github.com/gertab/ElixirST}} \ElixirST, assisting module construction in two ways:  
\begin{inparaenum}[(a)]
  \item it allows module designers to formalise the session endpoint protocol as a session type, and ascribe it to a public function;
  \item it implements a type-checker that verifies whether the body of a function respects the ascribed session type protocol specification. 
\end{inparaenum}

\vspace{-3mm}

\paragraph*{Contribution.} This paper 
validates the underlying type system on which the \ElixirST type-checker is built.
More concretely, in \Cref{sec:semantics} we formalise the runtime semantics of the \Elixir language fragment supported by \ElixirST as a labelled transition system (LTS), modelling the execution of a spawned handler interacting with a client within a session.
This operational semantics then allows us to prove \emph{session fidelity} for the \ElixirST type system in \Cref{sec:session-fidelity}.
In \Cref{sec:preliminaries} we provide the necessary background on the existing session type system from~\cite{DBLP:conf/agere/TaboneF21} to make the paper self-contained.

\section{Preliminaries}
\label{sec:preliminaries}
We introduce a core \Elixir subset and review the main typing rules for the \ElixirST type system~\cite{DBLP:conf/agere/TaboneF21}.

\subsection{The Actor Model}
\label{sec:actors}

\Elixir uses the actor concurrency model~\cite{DBLP:conf/ijcai/HewittBS73,DBLP:books/daglib/Agha90}.
It describes computation as a group of concurrent processes, called \emph{actors}, which do \emph{not} share any memory and interact exclusively via asynchronous messages.
Every actor is identified via a unique process identifier (\pid) which is used as the address when sending messages to a specific actor.
Messages are communicated asynchronously, and stored in the mailbox of the addressee actor. 
An actor is the only entity that can fetch messages from its mailbox, using mechanisms such as pattern matching.
Apart from sending and reading messages, an actor can also spawn other actors and obtain their fresh \pid as a result; this \pid can be communicated as a value to other actors via messaging.  



\subsection{Session Types}
\label{sec:formal-session-types}

The \ElixirST type system~\cite{DBLP:conf/agere/TaboneF21} assumes the standard expression types, including basic types, such as \texttype{boolean}, \texttype{number}, \texttype{atom} and \texttype{pid}, and inductively defined types, such as tuples ($\tupletype{T_1 \dotsspace T_n}$) and lists ($\listtype{T}$);
these already exist in the \Elixir language and they are dynamically checked.
It extends these with (binary) session types, which are used to statically check the message-passing interactions. 
%
\begin{equation*}
	\textrm{Expression types} \qquad T \Coloneqq \;\,
		\texttype{boolean} \; | \; \texttype{number} \; | \; \texttype{atom} \; | \; \texttype{pid} \; | \; \tupletype{T_1 \dotsspace T_n} \; | \; \listtype{T}
    \qquad\qquad\qquad\quad
\end{equation*}
\begin{align*}
	\textrm{Session types} && S \Coloneqq & \;\,
		 \stsessionbranch{?\texttt{l}_i \big(\widetilde{T_i} \big).S_i}{i \in I} && \textrm{Branch}
     & | & \;\, \stsessionrec{X}{S} && \textrm{Recursion} 
     \\
		  && | & \;\, \stsessionchoice{!\texttt{l}_i \big(\widetilde{T_i} \big).S_i}{i \in I} \; && \textrm{Choice} 
      & | & \;\, \stsessionrecvar{X} && \textrm{Variable}
      \\
		 && | & \;\, \stend && \textrm{Termination}
\end{align*}
The \emph{branching} construct, \stsessionbranch{?\texttt{l}_i \big(\widetilde{T_i} \big).S_i}{i \in I}, requires the code to be able to receive a message that is labelled by any one of the labels $\texttt{l}_i$, with the respective list of values of type $\widetilde{T_i}$ (where $\widetilde{T}$ stands for $T^1,\ldots,T^k$ for some $k\geq 0$), and then adhere to the continuation session type $S_i$. 
The \emph{choice} construct is its dual and describes the range and format of outputs the code is allowed to perform. 
In both cases, the labels need to be pairwise distinct.
Recursive types are treated {equi-recursively}~\cite{DBLP:books/daglib/0005958}, and used interchangeably with their unfolded counterparts.
For brevity, the symbols $\&$ and $\oplus$ are occasionally omitted for singleton options, \eg, $\stsessionchoice{!\textlabel{l}(\texttype{number}).S_1}{}$ is written as $!\textlabel{l}(\texttype{number}).S_1$; similarly $\stend$ may be omitted as well, \eg, $?\textlabel{l}()$ stands for $?\textlabel{l}().\stend$. 
The \emph{dual} of a session type $S$ is denoted as $\overline{S}$.

\subsection{Elixir Syntax}
\label{sec:syntax}

\begin{figure}[t]
	\hspace{-3mm}
	\begin{tabular}{l|l}
	\begin{minipage}{.5\linewidth} 
		$\begin{aligned}
		\textrm{Module} && M \Coloneqq & \;\, \stmodule{m}{\widetilde{P}}{\widetilde{D}} \\
		\textrm{Public fun.} && D \Coloneqq & \;\, K \quad B \quad \stdef{f}{\dualpid, \; \widetilde{x}}{t} \\
		\textrm{Private fun.} && P \Coloneqq & \;\, B \quad \stdefp{f}{\dualpid, \; \widetilde{x}}{t} \\
		\textrm{Type ann.} && B \Coloneqq & \;\, \stspecbig{f}{\widetilde{T}}{T} \\
		\textrm{Session ann.} && K \Coloneqq & \;\, \stsessionannotation{\stsessionrecvar{X} = S} \\ 
		&& | & \;\, \stdualannotation{\textsf{X}} 
		\\[7mm]
		\textrm{Expressions} && e \Coloneqq & \;\, \stexpvalue{w} \\ 
			&& | & \;\, \stexpnot{e} \; | \; \stexpbinaryop{e_1}{\diamond}{e_2} \\ 
			&& | & \;\, \listtypee{e_1}{e_2} \; | \; \tupletype{e_1 \dotsspace e_n} \\
		\textrm{Operators} && \diamond \Coloneqq & \;\, 
			< \; | \; > \; | \; <= \; | \; >= \; | \; == \\
			&& | & \;\, ! \! = \; | \; + \; | \; - \; | \; * \; | \; / \; | \; \stexpandonly  \; | \; \stexporonly
		\end{aligned}$
	\end{minipage}
		&
	\begin{minipage}{.5\linewidth} 
		$\begin{aligned}
			\textrm{Basic val.} && b \Coloneqq & \;\, \basicvalues{boolean} \; | \; \basicvalues{number} \; | \; \basicvalues{atom} \; | \; \basicvalues{pid} \\
			&& | & \;\, \listtype{} \\
			\textrm{Values} && v \Coloneqq & \;\, b \; | \; \listtypee{v_1}{v_2} \; | \; \tupletype{v_1 \dotsspace v_n} \\
			\textrm{Identifiers} && w \Coloneqq & \;\, b \; | \; x \\
			\textrm{Patterns} && p \Coloneqq & \;\, w \; | \; \listtypee{w_1}{w_2} \; | \; \tupletype{w_1 \dotsspace w_n} \\	
			\textrm{Terms} && t \Coloneqq & \;\, 
				e \\
				&& | & \;\, \stexpsequence{x}{t_1}{t_2} \\
				&& | & \;\, \stexpsend{\identifier}{\left\{ \atom{l}, e_1 \dotsspace e_n\right\}} \\
				&& | & \;\, 
					\begin{aligned}[t]
						&\stexpreceiveleft \\
						& \; \, \stexpreceiveright{\left\{ \atom{l}_i, p_i^1 \dotsspace p_i^n \right\}}{t_i}{i \in I}
					\end{aligned} \\
				&& | & \;\, \stexpfunction{f}{\identifier, \; e_1 \dotsspace e_n} \\ 
				&& | & \;\, \stexpcase{e}{p_i}{t_i}{i \in I} 
		\end{aligned}$
	\end{minipage}
	\end{tabular}
	\caption{Elixir syntax}
	\label{fig:elixir-syntax}
\end{figure}

\Elixir programs are organised as modules, \ie $ \stmodule{m}{\widetilde{P}}{\widetilde{D}}$.
Modules are defined by their name, $m$, and contain two sets of public $\widetilde{D}$ and private $\widetilde{P}$ functions, declared as sequences.
Public functions, $ \stdef{f}{\dualpid, \widetilde{x}}{t} $, are defined by the \elixircode{def} keyword, and can be called from any module.
In contrast, private functions, $ \stdefp{f}{\dualpid, \widetilde{x}}{t} $, can only be called from within the defining module.
Functions are defined by their name, $f$, and their body, $t$,  and parametrised by a sequence of \emph{distinct} variables, $\dualpid, \widetilde{x}$, the length of which, $|\dualpid, \widetilde{x}|$, is called the \emph{arity}.
As an extension to~\cite{DBLP:conf/agere/TaboneF21}, the first parameter ($y$), is reserved for the \pid of the dual process.  
Although a module may contain functions with the same name, their arity must be different.
%

In our formalisation, \Elixir function parameters and return values are assigned a type using the \annotation{spec} annotation, $f(\widetilde{T}) :: T$, describing the parameter types, $\widetilde{T}$, and the return type, $T$.
%
This annotation is already used by Dialyzer for success typing~\cite{DBLP:conf/ppdp/LindahlS06}.
In addition to this, we decorate public functions with session types, defined in \Cref{sec:formal-session-types}, to describe their side-effect protocol.
Public functions can be annotated directly using $\stsessionannotation {\stsessionrecvar{X} = S}$, or indirectly using the dual session type, $\stdualannotation{\stsessionrecvar{X}}$, where $\stsessionrecvar{X} = S$ is shorthand for $\stsessionrec{X}{S}$.
%

The body of a function consists of a term, $t$, which can take the form of an expression, a \texttt{let} statement, a send or receive construct, a case statement or a function call; see \Cref{fig:elixir-syntax}.  
In the case of the \texttt{let} construct, \stexpsequence{x}{t_1}{t_2}, the variable $x$ is a \emph{binder} for the variables in $t_2$ acting as a placeholder for the value that the subterm $t_1$ evaluates to.
We write $t_1; t_2$, as \emph{syntactic sugar} for \stexpsequence{x}{t_1}{t_2} whenever $x$ is not used in $t_2$.  
The \emph{send} statement, \stexpsend{\stexpvariable{x}}{\left\{ \atom{l}, e_1 \dotsspace e_n\right\}}, allows a process to send a message to the \pid stored in the variable $x$, containing a message $\left\{ \atom{l}, e_1 \dotsspace e_n\right\}$, where \atom{l} is the label.
The \emph{receive} construct, \stexpreceive{\left\{ \atom{l}_i, p_i^1 \dotsspace p_i^n \right\}}{t_i}{i \in I}, allows a process to receive a message tagged with a label that matches one of the labels $\atom{l}_i$ and a list of payloads that match the patterns $p_i^1 \dotsspace p_i^n$, branching to continue executing as $t_i$. 
Patterns, $p$, defined in \Cref{fig:elixir-syntax}, can take the form of a variable, a basic value, a tuple or a list (\eg \listtypee{x}{y}, where $x$ is the head and $y$ is the tail of the list).
The remaining constructs are fairly standard.
Variables in patterns $p_i^1 \dotsspace p_i^n$ employed by the \texttt{receive} and \texttt{case} statements are binders for the respective continuation branches $t_i$.
We assume standard notions of open (\ie $\fv{t} \neq \emptyset$) and closed  (\ie $\fv{t} = \emptyset$) terms and work up to alpha-conversion of bound variables.

\subsection{Type System}
\label{sec:behavioural-typing}

The session type system from~\cite{DBLP:conf/agere/TaboneF21} statically verifies that public functions within a module observe the communication protocols ascribed to them.
It uses three environments:
\begin{align*}
\text{Variable binding env.} & & 	\Gamma &\Coloneqq \emptyset \; \vert \; \Gamma, \; x : T \\
\text{Session typing env.} & & \Delta &\Coloneqq \emptyset \; \vert \; \Delta, \; \fun{f}{n} : S%
\\
\text{Function inf. env.} & & \Sigma &\Coloneqq 
\emptyset \; \vert \; \Sigma, \; \fun{f}{n} : \left\{ 
	\begin{aligned}
	&\texttt{params} = \widetilde{x}, \; \texttt{param\_types} = \widetilde{T}, \\ 
	&\texttt{body} = t, \; \texttt{return\_type} = T, \; \texttt{dual} = \dualpid
	\end{aligned}
	\right\}
\end{align*}

The \emph{variable binding} environment, $\Gamma$, maps (data) variables to basic types (\mbox{$x:T$}).
We write $\Gamma, x:T$ to extend $\Gamma$ with the new mapping $x:T$, where $x \notin \dom{\Gamma}$.
%
%
The \emph{session typing} environment, $\Delta$, maps function names and arity pairs to their session type ($\fun{f}{n} : S$).
%
%
If a function $\fun{f}{n}$ has a \emph{known} session type, then it can be found in $\Delta$, \ie $\Delta(\fun{f}{n}) = S$.
Each module has a static \emph{function information} environment, $\Sigma$, that holds information related to the function definitions.
%
For a function $f$, with arity $n$, $\Sigma(\fun{f}{n})$ returns the tail list of parameters ($\texttt{params}$) and their types ($\texttt{param\_types}$), the function's body ($\texttt{body}$), and its return type ($\texttt{return\_type}$).
In contrast to the original type system from~\cite{DBLP:conf/agere/TaboneF21}, $\Sigma(\fun{f}{n})$ also returns the variable name that represents the interacting process' \pid, \ie the option $\texttt{dual}$.
%
We assume that \emph{function information} environments, $\Sigma$, are \emph{well-formed}, meaning that all functions mapped  ($\fun{f}{n} \in \dom{\Sigma}$) observe the following condition requiring that the body of function $\fun{f}{n}$ is \emph{closed}:
	\begin{equation}
		\textbf{fv}\bigl(\Sigma(\fun{f}{n}).\texttt{body}\bigr) \setminus \bigl(\Sigma(\fun{f}{n}).\texttt{params}  \cup \Sigma(\fun{f}{n}).\texttt{dual}\bigr) = \emptyset
		\notag
	\end{equation}
  
Session typechecking is initiated by analysing an \Elixir module, rule \prooflabel{tModule}. 
%
A module is typechecked by inspecting each of its public functions, ascertaining that they correspond and fully consume the session types ascribed to them.
The rule uses three helper functions. 
%
The function $\functions{\widetilde{D}}$ returns a list of all function names (and arity) of the public functions ($\widetilde{D}$) to be checked individually.
The function $\sessions{\widetilde{D}}$ obtains a mapping of all the public functions to their expected session types stored in $\Delta$.
%
This ensures that when a function $f$ with arity $n$ executes, it adheres to the session type associated with it using either the \annotation{session} or \annotation{dual} annotations.
The helper function \detailsname populates the \emph{function information} environment ($\Sigma$) with details about all the \emph{public} ($\widetilde{D}$) and \emph{private} functions ($\widetilde{P}$) within the module.

%
\begin{figure}[H] 
	\centering


	\begin{center}
		\hspace*{-5mm}
		\alwaysNoLine
		\AxiomC{$ \Delta = \sessions{\widetilde{D}} \qquad \Sigma = \details{\widetilde{P}\,\widetilde{D}} \qquad  $}
		\UnaryInfC{$ \forall \fun{f}{n} \in \functions{\widetilde{D}} \cdot
		\begin{cases}
			\Delta(\fun{f}{n}) = S \qquad \Sigma\left(\fun{f}{n}\right) = \Omega\\
			\Omega \texttt{.params} = \widetilde{x} \qquad \Omega \texttt{.param\_types} = \widetilde{T} \\
			\Omega.\texttt{body} = t \qquad \Omega.\texttt{return\_type} = T \qquad \Omega.\texttt{dual} = \dualpid \\
			\tikzmarknode{tikz-typing-left}{}\typingsigma{\Delta}{\big(\dualpid : \texttype{pid}, \widetilde{x} : \widetilde{T}\big)}{\Sigma}{\dualpid}{S}{t}{T}{\texttt{end}}\tikzmarknode{tikz-typing-right}{}
		\end{cases}  $}
		\alwaysSingleLine
		\LeftLabel{\prooflabel{tModule}}
		\UnaryInfC{$ \typingmodule{\stmodule{m}{\widetilde{P}}{\widetilde{D}}} $}
		\DisplayProof
	\end{center}
	
	\begin{tikzpicture}[overlay,remember picture,>=stealth,nodes={align=left,inner ysep=1pt}]
		\draw[fill=xkcdLightOrange,opacity=0.30,blend mode=darken,draw=none] ([xshift=0mm,yshift=4.5mm]tikz-typing-left) rectangle ([xshift=1mm,yshift=-2.3mm]tikz-typing-right);
	\end{tikzpicture}


\end{figure}

For every public function \fun{f}{n} in $\functions{\widetilde{D}}$, \prooflabel{tModule} checks that its body adheres to it session type using  the \highlight{xkcdLightOrange}{highlighted} \emph{term typing} judgement detailed below:
%

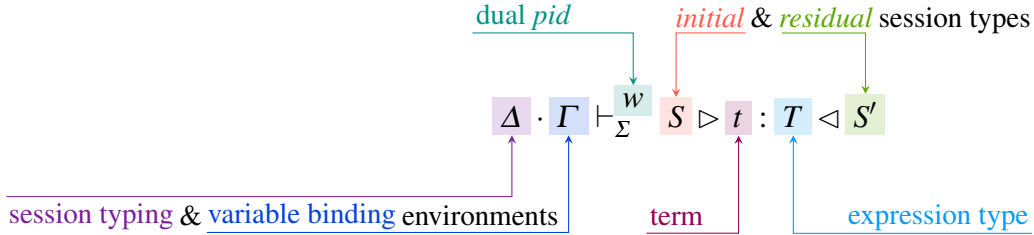
\begin{minipage}{0.98\columnwidth}
	\label{page:typing-session}
\vspace{12mm}
{\large
\begin{displaymath}
\hspace*{2.5cm}
\typingsigma{\tikzmarknode{delta}{\highlight{xkcdPurple}{$\Delta$}}}{\tikzmarknode{gamma}{\highlight{xkcdBlue}{$\Gamma$}}}{\Sigma}{\tikzmarknode{identifier}{\highlight{xkcdTeal}{$\identifier$}}}{\tikzmarknode{s}{\highlight{xkcdCoral}{$S$}}}{\tikzmarknode{t}{\highlight{xkcdDarkMagenta}{$t$}}}{\tikzmarknode{T}{\highlight{xkcdAzure}{$T$}}}{\tikzmarknode{sprime}{\highlight{xkcdLeafGreen}{$S'$}}} 
\end{displaymath}
}
\begin{tikzpicture}[overlay,remember picture,>=stealth,nodes={align=left,inner ysep=1pt}]
    \path (gamma.south) ++ (0,-2.3em) node[anchor=north east] (variableenvone){ environments};
	\path (variableenvone.west) ++ (0.5em, -0.1ex) node[anchor=east,color=xkcdBlue] (variableenvtwo) {variable binding};
	\path (variableenvtwo.west) ++ (0.5em, 0.1ex) node[anchor=east] (variableenvthree) {\&};
	\path (variableenvthree.west) ++ (0.5em, -0.13ex) node[anchor=east,color=xkcdPurple] (variableenvfour) {session typing};
    \draw [<-,color=xkcdBlue,xshift=0.7ex](gamma.south) |- ([xshift=0.7ex,color=blue]variableenvtwo.south west);
    \draw [<-,color=xkcdPurple](delta.south) |- ([color=blue,xshift=0.6ex,yshift=0.3ex]variableenvfour.north west);

	\path (s.north) ++ (12.6em,2.1em) node[anchor=south east] (sessiontypeone){session types};
	\path (sessiontypeone.west) ++ (0.5em, 0.4mm) node[anchor=east,color=xkcdLeafGreen] (sessiontypetwo) {\emph{residual}};
	\path (sessiontypetwo.west) ++ (0.5em, 0) node[anchor=east] (sessiontypethree) {\&};
	\path (sessiontypethree.west) ++ (0.5em, 0) node[anchor=east,color=xkcdCoral] (sessiontypefour) {\emph{initial}};
    \draw [<-,color=xkcdCoral](s.north) |- ([xshift=-0.7ex,yshift=-0.2ex,color=xkcdCoral]sessiontypefour.south east);
    \draw [<-,color=xkcdLeafGreen](sprime.north) |- ([xshift=1ex,yshift=-0.2ex,color=xkcdLeafGreen]sessiontypetwo.south west);

	\path (sessiontypefour.south east) ++ (-6em,-0.2ex) node[anchor=south east,color=xkcdTeal] (identifiertext){dual \pid};
    \draw [<-,color=xkcdTeal](identifier.north) |- ([xshift=0.5ex,color=red]identifiertext.south west);

	\path (variableenvone.south east) ++ (5em,0) node[anchor=south east,color=xkcdDarkMagenta] (termtext){term};
    \draw [<-,color=xkcdDarkMagenta](t.south) |- ([xshift=0.5ex,yshift=-0.38ex,color=red]termtext.south west);
	
	\path (variableenvone.south east) ++ (9.3em,-0.4ex) node[anchor=south west,color=xkcdAzure] (termtext){expression type};
    \draw [<-,color=xkcdAzure](T.south) |- ([xshift=-0.5ex,yshift=0.01ex,color=red]termtext.south east);

\end{tikzpicture}
\vspace{9.5mm}
\end{minipage}
%
This judgement
states that ``the term $t$ can produce a value of type $T$ after following an interaction protocol starting from the initial session type $S$ up to the residual session type $S'$, while interacting with a dual process with \pid identifier \identifier.
This typing is valid under some \emph{session typing} environment $\Delta$, \emph{variable binding} environment $\Gamma$ and \emph{function information} environment $\Sigma$.''
Since the \emph{function information} environment $\Sigma$ is static for the whole module (and by extension, for all sub-terms), it is left implicit in the term typing rules.
We consider four main rules, and relegate the rest to \Cref{sec:appendix-term-typing-rules}.
\begin{prooftree}
	\AxiomC{$  \forall i \in I \qquad \forall j \in 1..n \qquad \typingpattern{\identifier}{p_i^j}{T_i^j}{\Gamma_i^j} \qquad \typing{\Delta}{\big( \Gamma, \Gamma_i^1 \dotsspace \Gamma_i^n \big)}{\identifier}{S_i}{t_i}{T}{S'} $}
	\LeftLabel{\prooflabel{tBranch}}
	\UnaryInfC{$ \typing{\Delta}{\Gamma}{\identifier}{\stsessionbranch{? \textlabel{l}_i \big(\widetilde{T_i}\big) . S_i}{ i \in I}}{\stexpreceive{\left\{ \atom{l}_i, \widetilde{p_i} \right\}}{t_i}{i \in I }}{T}{S'} $}
\end{prooftree}
The \texttt{receive} construct is typechecked using the $\prooflabel{tBranch}$ rule.
%
It expects an (external) branching session type $\stsessionbranch{\dots}{}$, where each branch in the session type must match with a corresponding branch in the \texttt{receive} construct, where \emph{both} the labels ($\texttt{l}_i$) and payload types ($\widetilde{T}_i$) correspond.
The types within each \texttt{receive} branch are computed using the pattern typing judgement, $\typingpattern{\identifier}{p}{T}{\Gamma}$, which assigns types to variables present in patterns (see \Cref{sec:appendix-pattern-typing-rules}).
Each \texttt{receive} branch is then checked \wrt the common type $T$ and a common residual session type $S'$.
\begin{prooftree}
	\AxiomC{$ \exists i \in I $}	
	\AxiomC{$ \texttt{l} = \texttt{l}_i $}
	\AxiomC{$ \forall j \in 1..n $}
	\AxiomC{$ \typingexpression{\Gamma}{e_j}{T_i^j} $}
	\LeftLabel{\prooflabel{tChoice}}
	\QuaternaryInfC{$ \typing{\Delta}{\Gamma}{\identifier}{\stsessionchoice{! \textlabel{l}_i \big(\widetilde{T_i}\big) . S_i}{i \in I}}{\stexpsend{{\identifier}}{\left\{\atom{l}, e_1 \dotsspace e_n\right\}}}{ \left\{ \texttype{atom}, T_i^1 \dotsspace T_i^n\right\}}{S_i} $}
\end{prooftree}
The rule $\prooflabel{tChoice}$ typechecks the sending of messages.
This rule requires an internal choice session type $\stsessionchoice{\dots}{}$, where the label tagging the message to be sent must match with one of the labels ($\texttt{l}_i$) offered by the session choice.
The message payloads must also match with the corresponding types associated with the label  ($\widetilde{T}_i$ of $\texttt{l}_i$) stated via the expression typing judgement $\typingexpression{\Gamma}{e}{T}$ (see \Cref{sec:appendix-expression-typing-rules}).
The typing rule also checks the \pid of the addressee of the \texttt{send} statement which must match with the dual \pid ($\identifier$) states in the judgment itself
%
to ensure that messages are only sent to the correct addressee.
\begin{prooftree}
	\AxiomC{$ \Delta \left(\fun{f}{n}\right) = S \qquad \forall i \in 2..n \cdot \left\{ \typingexpression{\Gamma}{e_i}{T_i} \right\}$}
	\noLine
	\UnaryInfC{$ \Sigma \left(\fun{f}{n}\right) = \Omega \qquad \Omega.\texttt{return\_type} = T \qquad \Omega.\texttt{param\_types} = \widetilde{T} $}
	\LeftLabel{\prooflabel{tRecKnownCall}}
	\UnaryInfC{$ \typing{\Delta}{\Gamma}{\identifier}{S}{\stexpfunction{f}{\identifier, \; e_2 \dotsspace e_{n}}}{T}{\stend} $}
\end{prooftree}
Since public functions are decorated with a session type explicitly using the \annotation{session} (or \annotation{dual}) annotation, they are listed in $\dom{\Delta}$. 
Calls to public functions are typechecked using the \linebreak \prooflabel{tRecKnownCall} rule, which verifies that the expected initial session type is equivalent to the function's \emph{known} session type ($S$) obtained from the \textit{session typing} environment, \ie $\Delta \left(\fun{f}{n}\right) = S$.
Without typechecking the function's body, which is done in rule \prooflabel{tModule}, this rule ensures that the parameters have the correct types (using the expression typing rules). 
From the check performed in rule \prooflabel{tModule}, it can also safely assume that this session type $S$ is fully consumed, thus the residual type becomes $\stend$.
Rule \prooflabel{tRecKnownCall} also ensures that the \pid (\identifier) is preserved during a function call, by requiring it to be passed as a parameter and comparing it to the expected dual \pid (\ie $\typing{\Delta}{\Gamma}{\highlight{xkcdTeal}{\identifier}}{S}{\stexpfunction{f}{\highlight{xkcdTeal}{\identifier}, \dots}}{T}{\stend}$). 
%
%
\begin{prooftree}
	\AxiomC{$ \Sigma\left(\fun{f}{n}\right) = \Omega \qquad \qquad \fun{f}{n} \notin \dom{\Delta} \qquad \qquad \Omega.\texttt{dual} = \dualpid $}
	\noLine
	\UnaryInfC{$ \Omega.\texttt{params} = \widetilde{x} \quad\; \Omega.\texttt{param\_type} =  \widetilde{T} \quad\; \Omega.\texttt{body} = t \quad\; \Omega.\texttt{return\_type} = T  $}
	\noLine
	\UnaryInfC{$  $}	
	\noLine
	\UnaryInfC{$ \typing{\left(\Delta, \fun{f}{n} : S \right)}{\big(\Gamma, \dualpid : \texttype{pid}, \widetilde{x} : \widetilde{T}\big)}{\dualpid}{S}{t}{T}{S'} \qquad \;\; \forall  i \in 2..n \cdot \left\{ \typingexpression{\Gamma}{e_i}{T_i} \right\} $}
	\LeftLabel{\prooflabel{tRecUnknownCall}}
	\UnaryInfC{$ \typing{\Delta}{\Gamma}{\identifier}{S}{\stexpfunction{f}{\identifier, \; e_2 \dotsspace e_{n}}}{T}{S'} $}
\end{prooftree}
Contrastingly, a call to a (private) function, $\fun{f}{n}$, with an \emph{unknown} session type associated to it is typechecked using the \prooflabel{tRecUnknownCall} rule.
As in the other rule, it ensures that parameters have the correct types ($ \typingexpression{\Gamma}{e_i}{T_i}$).
However, it also analyses the function's body $t$ (obtained from $\Sigma$) with respect to the session type $S$ inherited from the initial session type of the call, 
Furthermore, this session type is appended to the \emph{session typing} environment $\Delta$ for future reference, \ie \mbox{$\Delta' = (\Delta, \fun{f}{n} : S)$} which allows it to handle recursive calls to itself; should the function be called again, rule \prooflabel{tRecKnownCall} is used thus bypassing the need to re-analyse its body.

\subsection{\Elixir System}


The \ElixirST provides a bespoke spawning function called \texttt{session/4} which allows the initiation of two concurrent processes executing in tandem as part of a session. 
This \texttt{session/4} function takes two pairs of arguments: two references of function names (that will be spawned), along with their list of arguments.
Its participant creation flow is shown in \Cref{fig:spawning}.  
Initially the actor (\textsf{pre-server}) is spawned, passing its \pid ($\pidvalue_{\textit{server}}$) to the second spawned actor (\textsf{pre-client}).
Then, \textsf{pre-client} relays back its \pid ($\pidvalue_{\textit{client}}$) to \textsf{pre-server}.
In this way, both actors participating in a session become aware of each other's \pids.
From this point onwards, the two actors execute their respective function to behave as the participants in the binary session; the first argument of each running function is initiated to the respective \pid of the other participant.
\Cref{fig:spawning} shows that the server process executes the body $t$, where it has access to the mailbox $\mathcal{M}$.
As it executes, messages may be sent or received (shown by the action $\alpha$) 
and stored in the mailbox $\mathcal{M}'$.
The specific working of these transitions is explained in the following section.

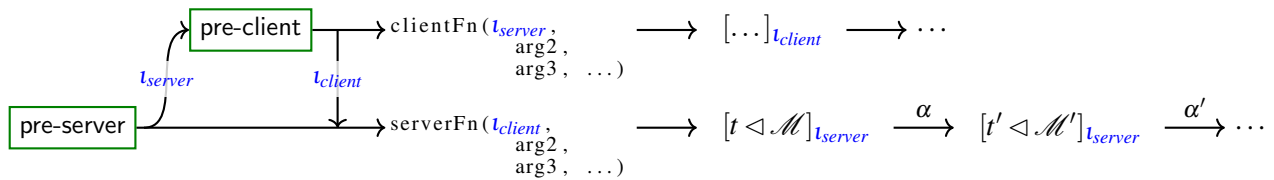
\begin{figure}[H]
	\hspace{-4mm}	
	\begin{tikzpicture}
		\node (pre-server) [spawn]{pre-server};
		\node (pre-client-above) [above=1.3cm of pre-server.east, anchor=east]{};
		\node (pre-client) [spawn, right=0.7cm of pre-client-above]{pre-client};
		\node (pre-client-client) [right=0.2cm of pre-client]{};
		\node (pre-client-client-below) [below=1.3cm of pre-client-client.center]{};
		
		\node (client-func) [right=5mm of pre-client-client]{};

		\node[xshift=1.5cm, yshift=-2.3mm] (client-func-code) at (client-func)
		{\begin{lstlisting}[numbers=none, frame=none, escapechar=`, basicstyle=\scriptsize\linespread{0.55}\selectfont, linewidth=3.4cm]
clientFn(`\footnotesize$\blueColor{\pidvalue_{\textit{server}}}$`, 
           arg2, 
           arg3, ...)
		\end{lstlisting}};

		\node (server-func) [below=1.3cm of client-func.west, anchor=west]{};
		\node[xshift=1.5cm, yshift=-2.3mm] (server-func-code) at (server-func) {\begin{lstlisting}[numbers=none, frame=none, escapechar=`, basicstyle=\scriptsize\linespread{0.55}\selectfont] 
serverFn(`\footnotesize$\blueColor{\pidvalue_{\textit{client}}}$`, 
           arg2, 
           arg3, ...)
		\end{lstlisting}};

		\draw[thick,->] (pre-server.east) to (server-func.west);
		\draw[thick,->] (pre-server.east) + (1mm, 0mm) to [out=0, in=180]  node[fill=white, text opacity=1,fill opacity=0.8]{\footnotesize$\blueColor{\pidvalue_{\textit{server}}}$} (pre-client.west);

		\draw[thick,->] (pre-client.east) to (client-func.west);
		\draw[thick,->] (pre-client-client.center) to node[fill=white, text opacity=1,fill opacity=0.8]{\footnotesize$\blueColor{\pidvalue_{\textit{client}}}$} (pre-client-client-below.north);

		\node (server-cont-server) [right=28mm of server-func]{};
		\node (server-cont-call-one) [right=8mm of server-cont-server]{\hspace*{2mm}$[t \lhd \mathcal{M}]_\text{\footnotesize$\blueColor{\pidvalue_{\textit{server}}}$}$};
		\draw[thick,->] (server-cont-server) to (server-cont-call-one);
		\node (client-cont-server) [right=28mm of client-func]{};
		\node (client-cont-client-dots) [right=8mm of client-cont-server]{\hspace*{2mm}$[\dots]_\text{\footnotesize$\blueColor{\pidvalue_{\textit{client}}}$}$};
		\draw[thick,->] (client-cont-server) to (client-cont-client-dots);

		\node (server-action) [right=-1mm of server-cont-call-one]{};
		\node (server-action-pid) [right=8mm of server-action]{\hspace*{2mm}$[t' \lhd \mathcal{M}']_\text{\footnotesize$\blueColor{\pidvalue_{\textit{server}}}$}$};
		\draw[thick,->] (server-action) -- (server-action-pid) node[above,text centered,midway,font=\footnotesize] {$\alpha$};
		

		\node (server-cont-server-dots) [right=-1mm of server-action-pid]{};
		\node (server-cont-client-dots) [right=8mm of server-cont-server-dots]{\dots};
		\draw[thick,->] (server-cont-server-dots) -- (server-cont-client-dots)  node[above,text centered,midway,font=\footnotesize] {$\alpha'$};
		\node (client-cont-server-dots) [right=-1mm of client-cont-client-dots]{};
		\node (client-cont-client-dots) [right=8mm of client-cont-server-dots]{\dots};
		\draw[thick,->] (client-cont-server-dots) to (client-cont-client-dots);

	\end{tikzpicture}
	\caption{Spawning two processes (\greenColor{green} boxes represent \emph{spawned} concurrent processes)}
	\label{fig:spawning}
\end{figure}

\section{Operational Semantics}
\label{sec:semantics}

We describe the operational semantics of the \Elixir language subset of \Cref{fig:elixir-syntax} as a \emph{labelled transition system} (LTS)~\cite{DBLP:journals/cacm/Keller76} describing how a handler process within a session executes while interacting with the session client, as outlined in \Cref{fig:roadmap}.   
The transitions  $\termsem{\alpha}{t}{t'}$ describes the fact that a handler process in state $t$ performs an execution step to transition to the new state $t'$ while producing action $\alpha$ as a side-effect.
External actions are visible by, and bear an effect on the client, whereas internal actions do not.
In our case, an action $\alpha$ can take the following forms: 
\begin{center}
	\begin{tikzpicture}
	\hspace*{-4mm}
	\node (actions) {
	\begin{minipage}{130mm}
	\begin{align*}
		\alpha \in \textsc{Act} \; \Coloneqq & \;\; 
		\actionsend{\pidvalue}{\tupletype{\atom{l}, \; \widetilde{v}}} & \text{Output message to $\pidvalue$ tagged as $\atom{l}$ with payload $\widetilde{v}$} 
    \\
		| & \;\; \actionreceive{\tupletype{\atom{l}, \; \widetilde{v}}} & \text{Input message tagged as $\atom{l}$ with payload $\widetilde{v}$} 
    \\
		| & \;\; \fun{f}{n} & \text{Call function $f$ with arity $n$}\\
		| & \;\; \tau & \text{Internal reduction step}
	\end{align*}
	\end{minipage}%
	};

	\draw [decorate,very thick,decoration = {brace,raise=-7mm,amplitude=2mm}] ([yshift=-8mm,xshift=1mm]actions.north east) --  ([yshift=-16mm,xshift=1mm]actions.north east) node[midway,left,xshift=21mm] () {external action};
	
	\draw [decorate,very thick,decoration = {brace,raise=-7mm,amplitude=2mm}] ([yshift=-19mm,xshift=1mm]actions.north east) --  ([yshift=-27mm,xshift=1mm]actions.north east) node[midway,left,xshift=21mm] () {internal action};
	\end{tikzpicture}
\end{center}
Both output and input actions constitute external actions that affect either party in a session; the type system from \Cref{sec:behavioural-typing} disciplines these external actions.
Internal actions, include \emph{silent} transition ($\tau$) and function calls ($\fun{f}{n}$); although the latter may be formalised as a silent action, the decoration facilitates our technical development. 
We note that, function calls can only transition subject to a well-formed \emph{function information} environment ($\Sigma$), which contains details about all the functions available in the module.
Since $\Sigma$ remains static during transitions, we leave it implicit in the transitions rules.

\begin{figure}[!b] 
	\begin{flushleft}
		\framebox[10px+\width]{$ \termsemm{\alpha}{\Sigma}{t}{t'} $}
	\end{flushleft}
	\vspace{-3.6\baselineskip}
	%
	\begin{center}
		\hspace*{5mm}
		\AxiomC{$ \termsem{\alpha}{t_1}{t_1'} $}	
		\LeftLabel{\prooflabell{rLet}{1}}
		\UnaryInfC{$ \termsem{\alpha}{\stexpsequence{x}{t_1}{t_2}}{\stexpsequence{x}{t_1'}{t_2}} $}
		\DisplayProof
		\hspace*{9mm}
		\AxiomC{$\vphantom{\termsem{\alpha}{t_1}{t_1'}} $}	
		\LeftLabel{\prooflabell{rLet}{2}}
		\UnaryInfC{$ \termsem{\tau}{\stexpsequence{x}{v}{t}}{t} \substitutionvx $}
		\DisplayProof
	\end{center}
%

	\begin{prooftree}
		\hspace*{-15mm}
		\AxiomC{$ \expressionsem{e_k}{e_k'} $}	
		\LeftLabel{\prooflabell{rChoice}{1}}
		\UnaryInfC{$ \termsem{\tau}{\stexpsend{\pidvalue}{\left\{\atom{l}, v_1 \dotsspace v_{k-1}, \; e_k \dotsspace e_n\right\}}}{\stexpsend{\pidvalue}{\left\{\atom{l}, v_1 \dotsspace v_{k-1}, \; e_k' \dotsspace e_n\right\}}}$}
	\end{prooftree}

	\begin{prooftree}
		\AxiomC{$\vphantom{\termsem{\alpha}{t_1}{t_1'}} $}	
		\LeftLabel{\prooflabell{rChoice}{2}}
		\UnaryInfC{$ \termsem{\actionsend{\pidvalue}{\tupletype{\atom{l}, \; v_1 \dotsspace v_n}}}{\stexpsend{\pidvalue}{\left\{\atom{l}, v_1 \dotsspace v_n\right\}}}{ \{ \atom{l}, v_1 \dotsspace v_n \} } $}
	\end{prooftree}


	\begin{prooftree}
		\AxiomC{$ \exists j \in I $}
		\AxiomC{$ \texttt{l}_j = \texttt{l} $}
		\AxiomC{$ \match{\widetilde{p_j}, \, v_1 \dotsspace v_n} = \sigma $}
		\LeftLabel{\prooflabel{rBranch}}
		\TrinaryInfC
    {$ 
    \termsem{
      \actionreceive{\tupletype{\atom{l}, \; v_1 \dotsspace v_n}}
      }
    {\stexpreceive{\left\{ \atom{l}_i, \widetilde{p_i} \right\}}{t_i}{i \in I}}
    { t_j \sigma } $}
	\end{prooftree}

	\begin{prooftree}
		\AxiomC{$ \expressionsem{e_k}{e_k'} $}
		\LeftLabel{\prooflabell{rCall}{1}}
		\UnaryInfC{$ \termsem{\tau}{\stexpfunction{f}{v_1 \dotsspace v_{k-1}, \; e_k \dotsspace e_n}}{\stexpfunction{f}{v_1 \dotsspace v_{k-1}, \; e_k' \dotsspace e_n}} $}
	\end{prooftree}

	\begin{prooftree}
		\AxiomC{$ \Sigma \left(\fun{f}{n}\right) = \Omega $}
		\AxiomC{$ \Omega.\texttt{body} = t $}
		\AxiomC{$ \Omega.\texttt{params} = x_2 \dotsspace x_n $}
		\AxiomC{$ \Omega.\texttt{dual} = y $}
		\LeftLabel{\prooflabell{rCall}{2}}
		\QuaternaryInfC{$ \termsem{\fun{f}{n}}{\stexpfunction{f}{\pidvalue, v_2 \dotsspace v_n}}{t \substitution{\pidvalue}{y} \substitution{v_2 \dotsspace v_n}{x_2 \dotsspace x_n}} $}
	\end{prooftree}

	\begin{prooftree}
		\AxiomC{$ \expressionsem{e}{e'} $}
		\LeftLabel{\prooflabell{rCase}{1}}
		\UnaryInfC{$ \termsem{\tau}{\stexpcase{e}{p_i}{t_i}{i \in I}}{\stexpcase{e'}{p_i}{t_i}{i \in I}} $}
	\end{prooftree}

	\begin{center}
		\AxiomC{$ \exists j \in I $}
		\AxiomC{$ \match{p_j, v} = \sigma $}
		\LeftLabel{\prooflabell{rCase}{2}}
		\BinaryInfC{$\termsem{\tau}{\stexpcase{v}{p_i}{t_i}{i \in I}}{t_j \sigma} $}
		\DisplayProof 
		\hspace*{12mm}
		\AxiomC{$ \vphantom{\exists j \in I \match{p_j, v} = \sigma} \expressionsem{e}{e'}$}	
		\LeftLabel{\prooflabel{rExpression}}
		\UnaryInfC{$ \termsem{\tau}{e}{e'} $}
		\DisplayProof
	\end{center}

	\caption{Term transition semantic rules}
	\label{fig:term-semantics}
\end{figure}

The transitions are defined by the \emph{term} transition rules listed  in \Cref{fig:term-semantics}.
%
Rules \prooflabell{rLet}{1} and \prooflabell{rLet}{2} deal with the evaluation of a \emph{let} statement, \stexpsequence{x}{t_1}{t_2} modelling a  \emph{call-by-value} semantic, where the first term $t_1$ has to transition fully to a value before being substituted for $x$ in $t_2$
denoted as \substitution{v}{x} (or \substitution{v_1, v_2}{x_1, x_2} for multiple substitutions).
%
The \emph{send} statement, \stexpsend{\pidvalue}{\left\{\atom{l}, e_1 \dotsspace e_n\right\}}, evaluates by first reducing each part of the message to a value from left to right.
This is carried out via rule \prooflabell{rChoice}{1} which produces no observable side-effects.
When the whole message is reduced to a tuple of values $\tupletype{\atom{l}, v_1 \dotsspace v_n}$, rule \prooflabell{rChoice}{2} performs the actual message sending operation.
This transition produces an action $\actionsend{\pidvalue}{\tupletype{\atom{l}, v_1 \dotsspace v_n}}$, where the message $\tupletype{\atom{l}, v_1 \dotsspace v_n}$ is sent to the interacting process, which has a \pid value of $\pidvalue$.
The operational semantics of the \emph{receive} construct, \stexpreceive{\left\{ \atom{l}_i,\widetilde{p_i} \right\}}{t_i}{i \in I}, is defined by rule \prooflabel{rBranch}.
When a message is received (\ie $\alpha = \, \actionreceive{\tupletype{\atom{l}, \; \widetilde{v}}}$), it is matched with a valid branch from the \emph{receive} construct, using the label $\atom{l}$. 
Should one of the labels match ($\exists {j\in I}$  such that $\atom{l}_j = \atom{l}$), the payload of the message ($\widetilde{v}$) is compared to the corresponding patterns in the selected branch ($\widetilde{p_j}$) using $\match{\widetilde{p_j}, \widetilde{v}}$.
If the values match with the pattern, the \matchname function (\Cref{def:pattern-matching}) produces the substitutions $\sigma$, 
mapping the matched variables in the pattern  $\widetilde{p_j}$ to values from $\widetilde{v}$.
This substitution $\sigma$ is then used to instantiate the free variables in continuation branch $t_j$.

\begin{definition}[Pattern Matching]
    \label{def:pattern-matching}
    The \matchname function pairs patterns with a corresponding value, resulting in a sequence of substitutions (called $\sigma$), \eg, $ \match{p, v} = \substitution{v_1}{x_1} \substitution{v_2}{x_2} = \substitution{v_1, v_2}{x_1, x_2}$. 
    Note that, a sequence of \matchname outputs are combined together, where the empty substitutions (\ie $ \substitutionsingle{\;} $) are ignored. 
		The match function builds a meta-list of substitutions, which is a different form of lists defined by the \Elixir syntax in \Cref{fig:elixir-syntax}.
	%
	\begin{align*}
        \hspace{12mm}
        \match{\widetilde{p}, \widetilde{v}} &\defsymbol \match{p_1, v_1} \dotsspace \match{p_n, v_n}\\[-1mm] 
		& \hspace*{45mm} \textit{where } \widetilde{p} = p_1 \dotsspace p_n \textit{ and } \widetilde{v} = v_1 \dotsspace v_n
		\\
        \match{p, v} &\defsymbol         
    \begin{cases}
        \substitutionsingle{\;} & p = b, v = b \textit{ and } p = v \\
        \substitutionvx & p = x \\
        \match{w_1, v_1}, \; \match{w_2, v_2} & p = \listtypee{w_1}{w_2}, v = \listtypee{v_1}{v_2} \\
        \match{w_1, v_1} \dotsspace \match{w_n, v_n} & p = \tupletype{w_1 \dotsspace w_n} \textit{ and}\\[-1mm] 
		&  \hspace{17mm} v = \tupletype{v_1 \dotsspace v_n} \hspace{12.9mm} \qedhere
    \end{cases}
\end{align*}
\end{definition}

\begin{figure}[t] 
	\begin{flushleft}
		\framebox[10px+\width]{$ \expressionsem{e}{e'} $}
	\end{flushleft}
	\vspace*{-9mm}
	\begin{center}
		\AxiomC{$ \expressionsem{e_1}{e_1'}   $}
		\LeftLabel{\prooflabell{reOperation}{1}}
		\UnaryInfC{$ \expressionsem{\stexpbinaryop{e_1}{\diamond}{e_2}}{e_1' \diamond e_2} $}
		\DisplayProof
		\hspace{1cm}
		\AxiomC{$ \expressionsem{e_2}{e_2'}   $}
		\LeftLabel{\prooflabell{reOperation}{2}}
		\UnaryInfC{$ \expressionsem{\stexpbinaryop{v_1}{\diamond}{e_2}}{v_1 \diamond e_2'} $}
		\DisplayProof

		\vspace*{7mm}
		\AxiomC{$ v = \stexpbinaryop{v_1}{\diamond}{v_2}$}
		\LeftLabel{\prooflabell{reOperation}{3}}
		\UnaryInfC{$ \expressionsem{\stexpbinaryop{v_1}{\diamond}{v_2}}{v} $}
		\DisplayProof
		\hspace{1mm}
		\AxiomC{$ \expressionsem{e}{e'}  $}
		\LeftLabel{\prooflabell{reNot}{1}}
		\UnaryInfC{$ \expressionsem{\stexpnot{e}}{e'} $}
		\DisplayProof
		\hspace{1mm}
		\AxiomC{$ v' = \neg v  $}
		\LeftLabel{\prooflabell{reNot}{2}}
		\UnaryInfC{$ \expressionsem{\stexpnot{v}}{v'} $}
		\DisplayProof

		\vspace*{7mm}
		\AxiomC{$ \expressionsem{e_1}{e_1'} $}
		\LeftLabel{\prooflabell{reList}{1}}
		\UnaryInfC{$ \expressionsem{\listtypee{e_1}{e_2}}{\listtypee{e_1'}{e_2}} $}
		\DisplayProof
		\hspace{1cm}
		\AxiomC{$ \expressionsem{e_2}{e_2'} $}
		\LeftLabel{\prooflabell{reList}{2}}
		\UnaryInfC{$ \expressionsem{\listtypee{v_1}{e_2}}{\listtypee{v_1}{e_2}} $}
		\DisplayProof

		\vspace*{7mm}
		\AxiomC{$ \expressionsem{e_k}{e_k'} $}
		\LeftLabel{\prooflabel{reTuple}}
		\UnaryInfC{$ \expressionsem{\tupletype{v_1 \dotsspace v_{k-1}, \; e_k \dotsspace e_n}}{\{v_1 \dotsspace, v_{k-1}, \; e_k' \dotsspace e_n\}} $}
		\DisplayProof
	\end{center}

	\mycomment{
		todo in write up add:
		For reTuple, if the paramaters are all expressions, then k = 1 and there are no values on the LHS
	}

	\caption{Expression reduction rules}
	\label{fig:expression-semantics}
\end{figure}

\begin{example}
	For the pattern $p_1 = \tupletype{x, 2, y}$ and the value tuple $v_1 =\tupletype{8, 2, \basicvalues{true}}$,
	 $ \match{p_1, v_1} = \sigma $ where $ \sigma = \substitution{8}{x} \substitution{\basicvalues{true}}{y} $ (written also as $ \sigma = \substitution{8, \basicvalues{true}}{x, y} $).
	However for pattern $p_2 = \tupletype{x, 2, \basicvalues{false}}$, the operation $ \match{p_2, v_1} $ fails, since $p_2$ expects a \basicvalues{false} value as the third element, but finds a \basicvalues{true} value instead.
	\qedhere
\end{example}

Using rule \prooflabell{rCall}{1} from \Cref{fig:term-semantics}, a function call is evaluated by first reducing all of its parameters to a value, using the expression reduction rules in \Cref{fig:expression-semantics}; again this models a call-by-value semantics.
Once all arguments have been fully reduced, \prooflabell{rCall}{2}, the implicit environment $\Sigma$ is queried for function $f$ with arity $n$ to fetch the function's parameter names and body.
This results in a transition to the function body with its parameters instantiated accordingly, $t \substitution{\pidvalue}{y} \substitution{v_2 \dotsspace v_n}{x_2 \dotsspace x_n}$, decorated by the function name, \ie $\alpha = \fun{f}{n}$.
Along the same lines a case construct first reduces the expression which is being matched using rule \prooflabell{rCase}{1}.
Then, rule \prooflabell{rCase}{2} matches the value with the correct branch, using the \matchname function, akin to \prooflabel{rBranch}.
Whenever a term consists solely of an expression, it silently reduces using \prooflabel{rExpression} using the expression reduction rules $\expressionsem{e}{e'}$ of \Cref{fig:expression-semantics}.
These are fairly standard.


\section{Session Fidelity}
\label{sec:session-fidelity}
We validate the static properties imposed by the \ElixirST type system~\cite{DBLP:conf/agere/TaboneF21}, overviewed in \Cref{sec:behavioural-typing}, by establishing a relation with the runtime behaviour of a typechecked \Elixir program, using the transition semantics defined in \Cref{sec:semantics}.
%
Broadly, we establish a form of \emph{type preservation}, which states that if a well-typed term transitions, the resulting term then remains well-typed~\cite{DBLP:books/daglib/0005958}.
However, our notion of type preservation,  needs to be stronger to also take into account
\begin{inparaenum}
  \item the side-effects produced by the execution; and 
  \item the progression of the execution with respect to protocol expressed as a session type.
\end{inparaenum}
Following the long-standing tradition in the session type community, these two aspects are captured by the refined preservation property called \emph{session fidelity}.
%
%
This property ensures that:
\begin{inparaenum}
    \item the communication action produced as a result of the execution of the typed process is one of the actions allowed by the current stage of the protocol; and that
    \item the resultant process following the transition is still well-typed \wrt the remaining part of the protocol that is still outstanding.
\end{inparaenum}
%
%

Before embarking on the proof for session fidelity, we prove an auxiliary proposition that acts as a sanity check for our operational semantics. 
We note that the operational semantics of \Cref{sec:semantics} assumes that only \emph{closed} programs are executed;
an \emph{open} program (\ie a program containing free variables) is seen as an incomplete program that cannot execute correctly due to missing information.
To this end, \Cref{prop:closed-term} ensures that a closed term \emph{remains closed} even after transitioning.

\begin{restatable}[Closed Term]{proposition}{closedtermprop}
    \label{prop:closed-term}
    If $ \fv{t} = \emptyset $ and $ \termsem{\alpha}{t}{t'} $, then $ \fv{t'} = \emptyset $
\end{restatable}

\begin{proof}
    By induction on the structure of $t$. 
    \iftoggle{shortversion}{%
    }{%
    Refer to \Cref{sec:appendix-proofs-prop-1} for details.
    }%
    \qedhere
\end{proof}


The statement of the session fidelity property relies on the definition of a partial function called \aftername (\Cref{def:after}), which takes a session type and an action as arguments and returns another session type as a result.
This function serves two purposes: 
\begin{inparaenum}[(a)]
    \item the function $\after{S, \alpha}$ is only defined for actions $\alpha$ that are (immediately) permitted by the protocol $S$, which allows us to verify whether a term transition step violated a protocol or not; and
    \item since $S$ describes the current stage of the protocol to be followed, we need a way to evolve this protocol to the next stage should $\alpha$ be a permitted action, and this is precisely $S'$, the continuation session type returned where $\after{S, \alpha} = S'$.
\end{inparaenum}

\begin{definition}[After Function]
    \label{def:after}
    The \aftername function is partial function defined for the following cases:
    \begin{align*}
        \after{S, \tau} &\defsymbol S \\
        \after{S, \fun{f}{n}} &\defsymbol S \\
        \after{\stsessionchoice{!\texttt{l}_i \big(\widetilde{T_i} \big).S_i}{i \in I}, \pidvalue !\left\{\texttt{l}_j, \widetilde{v}\right\}} &\defsymbol S_j \quad \text{ where } j \in I  \\ 
        \after{\stsessionbranch{?\texttt{l}_i \big(\widetilde{T_i} \big).S_i}{i \in I}, ? \left\{\texttt{l}_j, \widetilde{v}\right\}} &\defsymbol S_j \quad \text{ where } j \in I 
    \end{align*}

    \noindent
    This function is undefined for all other cases.
    The \aftername function is overloaded to range over  \emph{session typing} environments ($\Delta$) in order to compute a new \emph{session typing} environment given some action $\alpha$ and session type $S$:
    \begin{align*}
        \after{\Delta, \fun{f}{n}, S} &\defsymbol \Delta, \fun{f}{n}:S \\
        \after{\Delta, \alpha, S} &\defsymbol \Delta  \qquad\qquad \text{if } \alpha \neq \fun{f}{n} 
    \end{align*}
    Intuitively, when the action produced by the transition is $\fun{f}{n}$, the \emph{session typing} environment is extended by the new mapping $\fun{f}{n}:S$.
    For all other actions, the \emph{session typing} environment remains unchanged.
    \qedhere
\end{definition}

Recall that module typechecking  using rule \prooflabel{tModule} entails typechecking the bodies of all the public functions \wrt their ascribed session type, $\typingsigma{\Delta}{\big(\dualpid : \texttype{pid}, \widetilde{x} : \widetilde{T}\big)}{\Sigma}{\dualpid}{S}{t}{T}{S'}$ (where $S' = \texttt{end}$ for this specific case).
At runtime, a spawned client handler process in a session starts running the function body term $t$ where the parameter variables $y,\widetilde{x}$ are instantiated with the PID of the client, say $\pidvalue$, and the function parameter values, say $\widetilde{v}$, respectively, $t\substitution{\pidvalue}{y}\substitution{\widetilde{v}}{\widetilde{x}}$, as modelled in rule \prooflabell{rCall}{2} from \Cref{fig:term-semantics}. 
The instantiated function body is thus closed and can be typed \wrt an empty variable binding environment, $\Gamma = \emptyset$.
Session fidelity thus states that if a  closed term $t$ is well-typed, \ie
\begin{equation}
    \label{eq:session-fidelity-left}
    \typing{\Delta}{\emptyset}{\identifier}{S}{t}{\highlight{xkcdAzure}{$T$}}{S'}
\end{equation}
(where $S$ and $S'$ are initial and residual session types, respectively, and $T$ is the basic expression type) 
and this term $t$ transitions to a new term $t'$ with action $\alpha$, \ie
\begin{equation}
    \label{eq:session-fidelity-transition}
    \termsem{\alpha}{t}{t'}
\end{equation}
the new term $t'$ remains well-typed, \ie
\begin{equation}
    \label{eq:session-fidelity-right}
    \typing{\Delta'}{\emptyset}{\identifier}{S''}{t'}{\highlight{xkcdAzure}{$T$}}{S'}
\end{equation}
where the evolved $S''$ and $\Delta'$ are computed as $\after{S, \alpha} = S''$ and $\after{\Delta, \alpha, S} = \Delta'$.
This ensures that the base type of the term is preserved (note the constant type \highlight{xkcdAzure}{$T$} in \cref{eq:session-fidelity-left,eq:session-fidelity-right}). 
Furthermore, it ascertains that the term $t$ follows an interaction protocol starting from the initial session type $S$ up to the residual session type $S'$ (\cref{eq:session-fidelity-left}), since the updated session type $S''$ is defined for $\after{S, \alpha}$.

\begin{restatable}[Session Fidelity]{theorem}{presthm}
  \label{thm:session-fidelity}
  If $\typingsigma{\Delta}{\emptyset}{\Sigma}{\identifier}{S}{t}{T}{S'}$ and $\termsemm{\alpha}{\Sigma}{t}{t'}$, then there exists some $S''$ and $\Delta'$, such that $\typingsigma{\Delta'}{\emptyset}{\Sigma}{\identifier}{S''}{t'}{T}{S'}$ for $\after{S, \alpha} = S''$ and $\after{\Delta, \alpha, S} = \Delta'$
\end{restatable}
\begin{proof}
  By induction on the typing derivation $\typingsigma{\Delta}{\emptyset}{\Sigma}{\identifier}{S}{t}{T}{S'}$.
  \iftoggle{shortversion}{%
  }{%
    Refer to \Cref{sec:appendix-proofs-thm}.
  }
\end{proof}

\begin{figure}
  \begin{center}
    \begin{tikzpicture}
        \node[] (t) {$t_1$};
        \node[below=5mm of t] (t-S) {$\highlight{yellow}{\small$S_1$}$};
        \node[right=21mm of t] (tprime) {$t_2$};
        \node[below=5mm of tprime] (tprime-S) {\small$\after{S_1, \alpha_1} = S_2$};
        \node[right=21mm of tprime] (tdots) {$\dots$};
        \node[right=21mm of tdots] (tv) {$v$};
        \node[below=5mm of tv] (tv-S) {\small$\after{S_n, \alpha_n} = \highlight{yellow}{\small$\stend$}$};
        \node[right=10mm of tv] (end) {};
    
        \draw[->](t.east) -- (tprime.west) node[above,text centered,midway,font=\footnotesize] {$\alpha_1$};
        \draw[->](tprime.east) -- (tdots.west) node[above,text centered,midway,font=\footnotesize] {$\alpha_2$};
        \draw[->](tdots.east) -- (tv.west) node[above,text centered,midway,font=\footnotesize] {$\alpha_n$};
        \draw[->, negated](tv.east) -- (end.west);
        
        \draw[densely dotted](t) -- (t-S);
        \draw[densely dotted](tprime) -- (tprime-S);
        \draw[densely dotted](tv) -- (tv-S);
    \end{tikzpicture}
    \end{center}
    \caption{Repeated applications of \emph{session fidelity}}
    \label{fig:session-fidelity-in-use}
\end{figure}
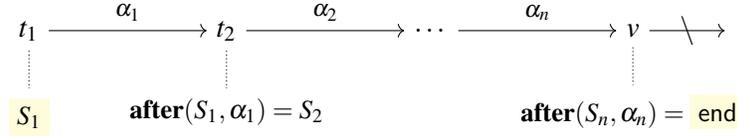

As shown in \Cref{fig:session-fidelity-in-use}, by repeatedly applying  \Cref{thm:session-fidelity}, we can therefore conclude that all the (external) actions generated as a result of a computation (\ie sequence of transition steps) must all be actions that follow the protocol described by the initial session type. 
Since public functions are always typed with a residual session type $\stend$, certain executions could reach the case where the outstanding session is updated to  $\stend$ as well, \ie  $\after{S_n, \alpha_n} = \stend$.
In such a case, we are guaranteed that the term will not produce further side-effects, as in the case of \Cref{fig:session-fidelity-in-use} where the term is reduced all the way down to some value, $v$.

\begin{example}
We consider a concrete example to show the importance of session fidelity.
The function called $\fun{\texttt{pinger}}{1}$  is able to send \textlabel{ping} and receive \textlabel{pong} repeatedly.

\begin{minipage}{0.49\linewidth}
\begin{lstlisting}[language=elixir] 
@session "X = !ping().?pong().X"
def pinger(pid) do
  x = send(pid, {:ping})

  receive do
    {:pong} -> IO.puts("Received pong.")
  end
  pinger(pid)
end
\end{lstlisting}
\end{minipage}

This function adheres to the following protocol:
\begin{equation*}
    \hspace*{-2mm}
    \stsessionrecvar{X} = \; !\textlabel{ping}().?\textlabel{pong}().\stsessionrecvar{X}
\end{equation*}

A process evaluating the function \texttt{pinger} executes by first sending a message containing a \textlabel{ping} label to the interacting processes' \pid ($\pidvalue_{pong}$), as shown below.

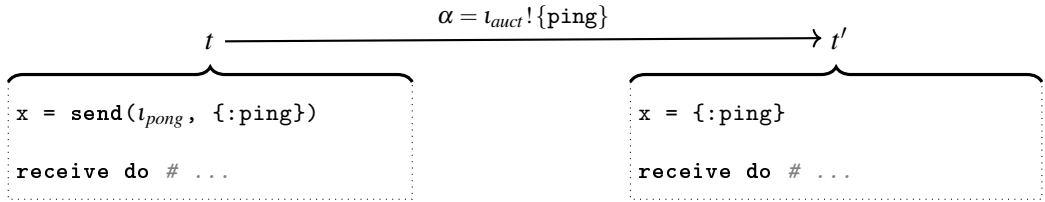
\vspace{2mm}
\begin{tikzpicture}
\hspace*{-2mm}
\node (first) {
\begin{minipage}{6cm}
\begin{lstlisting}[language=elixir,numbers=none,frame=single] 
x = send((*$\pidvalue_{pong}$*), {:ping})

receive do # ...
\end{lstlisting}
\end{minipage}
};

\node[right=20mm of first] (second) {
\begin{minipage}{6cm}
\begin{lstlisting}[language=elixir,numbers=none,frame=single] 
x = {:ping}

receive do # ...
\end{lstlisting}
\end{minipage}
};

\draw [decorate,very thick,decoration = {brace,amplitude=2mm}] ([xshift=7mm]first.north west) --  ([xshift=-2mm]first.north east);
\draw [decorate,very thick,decoration = {brace,amplitude=2mm}] ([xshift=7mm]second.north west) --  ([xshift=-1mm]second.north east);
\draw[dotted] ([xshift=7mm]first.north west) -- ([xshift=7mm,yshift=3mm]first.south west) -- ([xshift=-2mm,yshift=3mm]first.south east) --  ([xshift=-2mm]first.north east);
\draw[dotted] ([xshift=7mm]second.north west) -- ([xshift=7mm,yshift=3mm]second.south west) -- ([xshift=-1mm,yshift=3mm]second.south east) --  ([xshift=-1mm]second.north east);

\node (t) [above=3mm of first.north,xshift=2.5mm]{$t$};
\node (tprime) [above=3mm of second.north,xshift=3.5mm]{$t'$};
\draw[->,thick](t.east) -- (tprime.west) node[above,text centered,midway,font=\footnotesize] {$\alpha = \pidvalue_{auct} !\left\{\texttt{ping}\right\}$};
\end{tikzpicture}

\noindent
As the process evaluates, the initial term $t$ transitions to $t'$,  where it sends a message as a side-effect.
This side-effect is denoted as an action $\alpha$, where $\alpha = \pidvalue_{pong} !\left\{\texttt{ping}\right\}$. 
By the \fullref{def:after}, $\stsessionrecvar{X}$ evolves to a new session type $\stsessionrecvar{X'}$:
\begin{equation*}
    \stsessionrecvar{X'} = \after{!\textlabel{ping}().?\textlabel{pong}().\stsessionrecvar{X}, \alpha} = \; 
    ?\textlabel{pong}().\stsessionrecvar{X}
\end{equation*} 

For $t'$ to remain well-typed, it must now match with the evolved session type $\stsessionrecvar{X'}$, where it has to be able to receive a message labelled \textlabel{pong}, before recursing.
Although the process keeps executing indefinitely, by the session fidelity property, we know that each step of execution will be in line with the original protocol.
\qedhere
\end{example}



\section{Related Work}
\label{sec:related-work}

In this section, we compare \ElixirST with other type systems and implementations.

\paragraph*{Type Systems for \Elixir}

Cassola~\etal~\cite{cassola2022101077,Cassola2020} presented a gradual type system for \Elixir.
It statically typechecks the functional part of \Elixir modules, using a gradual approach, where some terms may be left with an unknown expression type.
%
%
In contrast to \ElixirST, Cassola~\etal analyse directly the unexpanded \Elixir code which results in more explicit typechecking rules.
Also, they focus on the static type system without formulating the operational semantics.

Another static type-checker for \Elixir is \textit{Gradient}~\cite{github/gradient}.
It is a wrapper for its \Erlang counterpart tool and takes a similar approach to~\cite{cassola2022101077}, where gradual types are used. 
Another project, \textit{TypeCheck}~\cite{github/Qqwy/typecheck}, adds dynamic type validations to \Elixir programs.
\textit{TypeCheck} performs runtime typechecking by wrapping checks around existing functions.
\textit{Gradient} and \textit{TypeCheck} are provided as an implementation only, without any formal analysis.
In contrast to \ElixirST, the discussed type-checkers~\cite{cassola2022101077,github/gradient,github/Qqwy/typecheck} analyse the sequential part of the \Elixir language omitting any checks related to message-passing between processes.

Some implementations aim to check issues related to message-passing.
Harrison~\cite{DBLP:conf/erlang/Harrison18} statically checks Core \Erlang for such issues.
For instance, it detects orphan messages (\ie messages that will never be received) and unreachable receive branches.
Harrison~\cite{DBLP:conf/erlang/Harrison19} extends~\cite{DBLP:conf/erlang/Harrison18} to add analyse Erlang/OTP behaviours (\eg, \texttt{gen\_server}, which structures processes in a hierarchical manner) by injecting runtime checks in the code. 
Compared to our work, \cite{DBLP:conf/erlang/Harrison18,DBLP:conf/erlang/Harrison19} perform automatic analysis of the implementation, however they do not verify communication with respect to a general protocol (\eg, session types).

Another type system for \Erlang was presented Svensson~\etal~\cite{DBLP:conf/erlang/SvenssonFE10}.
Their body of work covers a larger subset of Erlang to what would be its equivalent in Elixir covered by our work. 
Moreover, its multi-tiered semantics captures an LTS defined over systems of concurrent actors.
Although we opted for a smaller subset, we go beyond the pattern matching described by Svensson~\etal since we perform a degree of typechecking for base types (\eg in the premise of \prooflabel{tBranch}).



\paragraph*{Session Type Systems.}
Closest to our work is~\cite{DBLP:conf/coordination/MostrousV11}, where Mostrous and Vasconcelos introduced session types to a fragment of Core \Erlang, a dynamically typed language linked to \Elixir.
Their type system tags each message exchanged with a unique reference.
This allows multiple sessions to coexist, since different messages could be matched to the corresponding session, using correlation sets.
Mostrous and Vasconcelos take a more theoretic approach, so there is no implementation for~\cite{DBLP:conf/coordination/MostrousV11}.
Their type system guarantees \emph{session fidelity} by inspecting the processes' mailboxes where, at termination, no messages should be left unprocessed in their mailboxes. 
Our work takes a more limited but pragmatic approach, where we introduce session types for functions within a module.
Furthermore, we offer additional features, including variable binding (\eg, in let statements),  expressions (\eg, addition operation), inductive types (\eg, tuples and lists), infinite computation via recursion and explicit protocol definition.

A session-based runtime monitoring tool for python was initially presented by Neykova and \linebreak Yoshida~\cite{DBLP:journals/corr/NeykovaY16,DBLP:journals/corr/NeykovaY14}.
They use the Scribble \cite{DBLP:conf/icdcit/HondaMBCY11} language to write \emph{multiparty} session type (MPST)~\cite{DBLP:journals/jacm/HondaYC16} protocols, which are then used to monitor the processes' actions.
Different processes are ascribed a role (defined in the MPST protocol) using function decorators (akin to our function annotations).
%
Similar to~\cite{DBLP:journals/corr/NeykovaY16,DBLP:journals/corr/NeykovaY14}, Fowler~\cite{DBLP:journals/corr/Fowler16} presented an MPST implementation for \Erlang.
This implementation uses Erlang/OTP behaviours (\eg, \texttt{gen\_server}), which take into account Erlang's \emph{let it crash} philosophy, where processes may fail while executing.
In contrast, although our work accepts a more limited language, \ElixirST provides static guarantees where issues are flagged at pre-deployment stages, rather than flagging them at runtime.
%

Scalas and Yoshida~\cite{DBLP:conf/ecoop/ScalasY16} applied binary session types to the Scala language, where session types are abstracted as Scala classes.
Session fidelity is ensured using Scala's compiler, which complains if an implementation does not follow its ascribed protocol.
Linearity checks are performed at runtime, which ensure that an implementation fully exhausts its protocol exactly once.
%
%
%
Bartolo Burl{\`o}~\etal~\cite{DBLP:conf/ecoop/BurloFS21} extended the aforementioned work~\cite{DBLP:conf/ecoop/ScalasY16}, to monitor one side of an interaction statically and the other side dynamically using runtime monitors.

Harvey~\etal~\cite{DBLP:conf/ecoop/00020DG21} presented a new actor-based language, called EnsembleS, which offers session types as a native feature of the language.
EnsembleS statically verifies implementations with respect to session types, while still allowing for adaptation of \emph{new} actors at runtime, given that the actors obey a known protocol.
Thus, actors can be terminated and discovered at runtime, while still maintaining static correctness.

There have been several binary~\cite{DBLP:conf/icfp/JespersenML15,DBLP:journals/corr/abs-1909-05970} and multiparty~\cite{DBLP:conf/coordination/LagaillardieNY20,DBLP:conf/coordination/CutnerY21} session type implementations for Rust.
These implementations exploit Rust's affine type system to guarantee that channels mirror the actions prescribed by a session type.
%
Padovani~\cite{DBLP:journals/jfp/Padovani17} created a binary session type library for OCaml to provide static communication guarantees. 
This project was extended~\cite{DBLP:journals/pacmpl/MelgrattiP17} to include dynamic contract monitoring which flags violations at runtime.
%
The approaches used in the Rust and OCaml implementations rely heavily on type-level features of the language, which do not readily translate to the dynamically typed \Elixir language.



\section{Conclusion}
\label{sec:conclusion}


In this work we established a correspondence between the \ElixirST type system~\cite{DBLP:conf/agere/TaboneF21} and the runtime behaviour of a client handler running an \Elixir module function that has been typechecked \wrt its session type protocol.
In particular, we showed that this session-based type system observes the standard \emph{session fidelity} property, meaning that processes executing a typed function \emph{always} follow their ascribed protocols at runtime.
This property provides the necessary underlying guarantees  to attain various forms of communication safety, whereby should two processes following mutually compatible protocols (\eg $S$ and its dual $\bar{S}$), they avoid certain communication errors (\eg, a send statement without a corresponding receive construct).
An extended version of this work can be found in the technical report~\cite{TaboneFrancalanzaTechReport}.

\paragraph*{Future work.}
There are a number of avenues we intend to pursue.
One line of investigation is the augmentation of protocols that talk about multiple entry points to a module perhaps from the point of view of a client that is engaged in multiple sessions at one time, possibly involving multiple modules. 
The obvious starting points to look at here are the well-established notions of multiparty session types~\cite{DBLP:journals/jacm/HondaYC16,DBLP:conf/pldi/ScalasYB19} or the body of work on intuitionistic session types organising 
processes hierarchically~\cite{DBLP:journals/pacmpl/BalzerP17,pruiksmaPfenning2022:JFP}.   
Another natural extension to our work would be to augment our session type protocol in such a way to account for process failure and supervisors, which is a core part of the \Elixir programming model. 
For this, we will look at previous work on session type extensions that account for failure~\cite{DBLP:conf/ecoop/00020DG21}.
Finally, we also plan to augment our session typed protocols to account for resource usage and cost, along the lines of \citep{DBLP:conf/lics/Das0P18,DBLP:journals/corr/FrancalanzaVH14}.

\paragraph*{Acknowledgements.}
We thank Simon Fowler and Matthew Alan Le Brun for their encouragement and suggestions for improvement on earlier versions of the work.
This work has been supported by the  MoVeMnt project (No:\,217987)  of  the Icelandic Research Fund and the BehAPI project funded by the EU H2020 RISE of the Marie Sk{\l}odowska-Curie action (No:\,778233).

\bibliographystyle{eptcs}
\bibliography{refs}

\newpage
  
\appendix

\part*{Appendix}

\section{Additional Definitions}
\label{sec:appendix-definitions}
In this appendix, we formalise some auxiliary definitions that were used in \Cref{sec:preliminaries,sec:semantics,sec:session-fidelity}.

\iftoggle{shortversion}{%
}{%
\begin{definition}[Free Variables]
    \label{def:free-variables}
    The set of free variables is defined inductively as:
    \begin{align*}
    \hspace{20mm} \fv{e} &\defsymbol         
    \begin{cases}
        \{x\} & e = x \\
        \emptyset & e = b \\
        \fv{e_1} \cup \fv{e_2} & e = \stexpbinaryop{e_1}{\diamond}{e_2} \text{ or } e = \listtypee{e_1}{e_2} \\
        \fv{e'} & e = \stexpnot{e'} \\
        \cup_{i \in 1..n}\fv{e_i} & e = \tupletype{e_1 \dotsspace e_n}
    \end{cases} \\
    \fv{t} &\defsymbol        
    \begin{cases}
        \fv{t_1} \cup (\fv{t_2} \setminus \{x\}) & t = ( \stexpsequence{x}{t_1}{t_2} ) \\
        \cup_{i \in 1..n}\fv{e_i} \cup \fv{\identifier} & t = \stexpsend{\identifier}{\left\{\atom{l}, e_1 \dotsspace e_n\right\}} \\
        \cup_{i \in I}[\fv{t_i} \setminus \vars{\widetilde{p_i}}] & t = 
        \stexpreceive{\left\{ \atom{l}_i, \widetilde{p_i} \right\} }{t_i}{i \in I} \\
        \cup_{i \in 2..n}\fv{e_i} \cup \fv{\identifier} & t = \stexpfunction{f}{\identifier, e_2 \dotsspace e_n} \\
        \cup_{i \in I}[\fv{t_i} \setminus \vars{p_i}] \cup \fv{e} & t = \stexpcase{e}{p_i}{t_i}{i \in I}  \hspace{22.2mm} \qedhere
    \end{cases}
\end{align*}
\end{definition}


\begin{definition}[Bound Variables]\ \\
    \label{def:bound-variables}
    \begin{equation*}
        \hspace{15mm} \bv{t} \defsymbol         
    \begin{cases}
        \emptyset & t = e \text{ or } t = \stexpsend{\identifier}{\left\{\atom{l}, \widetilde{e} \right\}} \text{ or } t = \stexpfunction{f}{\widetilde{e} } \\
        \{x\} \cup \bv{t_1} \cup \bv{t_2} & t = ( \stexpsequence{x}{t_1}{t_2} ) \\
        \cup_{i \in I}[\bv{t_i} \cup \vars{\widetilde{p_i}}] & t = 
        \stexpreceive{\left\{ \atom{l}_i, \widetilde{p_i} \right\} }{t_i}{i \in I} \\
        \cup_{i \in I}[\bv{t_i} \cup \vars{p_i}] & t = \stexpcase{e}{p_i}{t_i}{i \in I}
        \hspace{37.4mm} \qedhere
    \end{cases}
    \end{equation*}
\end{definition}


\begin{definition}[Variable Substitution]\ \label{def:sub}~
    %
    \begin{align*}
         e \substitutionvx &\defsymbol  
        \begin{cases}
            v & e = x \\
            y & e = y, \; y \neq x \\
            b & e = b \\
            \stexpbinaryop{e_1 \substitutionvx}{\diamond}{e_2 \substitutionvx} & e = \stexpbinaryop{e_1}{\diamond}{e_2} \\
            \stexpnot{(e'\substitutionvx)} & e = \stexpnot{e'} \\
            \listtypee{e_1\substitutionvx}{e_2\substitutionvx} & e = \listtypee{e_1}{e_2} \\
            \tupletype{e_1\substitutionvx \dotsspace e_n\substitutionvx} \qquad \qquad \qquad \qquad & e = \tupletype{e_1 \dotsspace e_n}
        \end{cases} \\
        \hspace{9mm}
        t \substitutionvx &\defsymbol  
        \begin{cases}
            \stexpsend{\identifier\substitutionvx}{\{:\textlabel{l} , \; e_1 \substitutionvx \dotsspace e_n \substitutionvx \}} & t = \stexpsend{\identifier}{\{:\textlabel{l} , \; e_1 \dotsspace e_n \}} \\
            \stexpreceive{ \left\{ \textlabel{l}_i, \widetilde{p_i} \right\} }{t_i\substitutionvx}{i \in I} & t  = \stexpreceive{ \left\{ \textlabel{l}_i, \widetilde{p_i} \right\} }{t_i}{i \in I} \\
            \stexpfunction{f}{e_1\substitutionvx \dotsspace e_n\substitutionvx} & t = \stexpfunction{f}{e_1 \dotsspace e_n} \\
            \stexpcase{e\substitutionvx}{p_i}{t_i\substitutionvx}{i \in I} & t = \stexpcase{e}{p_i}{t_i}{i \in I} \\ 
            \stexpsequence{y}{t_1\substitutionvx}{t_2\substitutionvx} & t = (\stexpsequence{y}{t_1}{t_2}), \; x \neq y, \; y \neq v
            \hspace{18.93mm}\qedhere
        \end{cases}
    \end{align*}
\end{definition}

}%


\begin{definition}[Type]\ \label{def:typeof}~
    \begin{align*}
        \typeof{\basicvalues{boolean}} &\defsymbol \texttype{\texttype{boolean}} &
        \typeof{\basicvalues{number}} &\defsymbol \texttype{\texttype{number}} \\
        \typeof{\basicvalues{atom}} &\defsymbol \texttype{\texttype{atom}} &
        \typeof{\pidvalue} &\defsymbol \texttype{pid}, \text{ where $\pidvalue$ is a $\basicvalues{pid}$ instance}
        \qedhere
    \end{align*}
\end{definition}

\iftoggle{shortversion}{%
}{%
\begin{definition}[Variable Patterns]
    \label{def:pattern-variables}
    Computes an ordered set of variables from a given pattern $p$. 
    %
    \begin{align*}
        \hspace{33.783mm}\vars{\widetilde{p}} &\defsymbol \vars{p_1 \dotsspace p_n} \defsymbol \vars{p_1} \cup \dots \cup \vars{p_n} \\
        \vars{p} &\defsymbol         
    \begin{cases}
        \emptyset & p = b \\
        \{x\} & p = x \\
        \vars{w_1} \cup \vars{w_2} & p = \listtypee{w_1}{w_2} \\
        \cup_{i \in 1..n} \vars{w_i} & p = \tupletype{w_1 \dotsspace w_n} \hspace{33.783mm} \qedhere
    \end{cases}
    \end{align*}
\end{definition}
}%





\begin{definition}[Function Details]
	\label{def:details}
	We can extract function details (\ie \texttt{params}, \texttt{body}, \texttt{param\_types}, \texttt{return\_type}, \texttt{dual}) from a list of functions ($\widetilde{Q}$) and build a mapping, using set-comprehension, as follows. 
	The list of functions ($\widetilde{Q}$) may consist of public ($D$) and private ($P$) functions.
\begin{equation*}
	\details{\widetilde{Q}} \defsymbol 
	\left\{ 
	\;\;
	\fun{f}{n} :  
	\left[ 
	\begin{aligned}
		&\texttt{dual} = \dualpid, \; \texttt{params} = \widetilde{x}, \\
		&\texttt{param\_types} = \widetilde{T}, \\
		&\texttt{return\_type} = T, \; \texttt{body} = t
	\end{aligned}
	\right]
	\quad \middle | \quad
	\left[
	\begin{aligned}
		& [\stsessionannotation{S}] \\
		& \stspec{f}{\texttype{pid},\widetilde{T}}{T} \\
		& \texttt{def}[\texttt{p}] \; f \! \left( \dualpid,\widetilde{x} \right) \texttt{do} \; t \; \texttt{end}
	\end{aligned}
	\right]
	\in \widetilde{Q}
	\;\; \right\}
\end{equation*}

\end{definition}


\begin{definition}[Functions Names and Arity]
	\label{def:functions}
    This definition takes the set of all public function ($\widetilde{D}$) as input, and returns a set of all public function names and their arity.
%
	\begin{equation*}
		\functions{\widetilde{D}} \defsymbol 
		\left\{ 
		\;\;
		\fun{f}{n}
		\quad \middle | \quad
		\left[
		\begin{aligned}
			& \texttt{@session} \dots; \; \texttt{@spec} \dots \\
			& \stdef{f}{\dualpid,\, x_2 \dotsspace x_n}{t}
		\end{aligned}
		\right]
		\in \widetilde{D}
		\;\; \right\}
	\end{equation*}
    \qedhere
\end{definition}


\begin{definition}[All Session Types]
	\label{def:session}
	The function $\sessions{\widetilde{D}}$, returns the session type corresponding to each annotated public function.
%
%
	\begin{equation*} 
		\sessions{\widetilde{D}} \defsymbol 
		\left\{ 
		\;\;
		\fun{f}{n} : S
		\quad \middle | \quad
		\left[
		\begin{aligned}
			& \texttt{@session } ``S"; \; \texttt{@spec} \dots \\
			& \stdef{f}{\dualpid, \, x_2 \dotsspace x_{n}}{t}
		\end{aligned}
		\right]
		\in \widetilde{D}
		\;\; \right\}
	\end{equation*}
	In case the \annotation{dual} annotation is used instead of \annotation{session}, the dual session type is computed automatically.
	\qedhere
\end{definition}

\newpage

\section{Type System Rules}
\label{sec:appendix-type-system-rules}
In this appendix, we present the full typing rules of the type system, adapted from~\cite{DBLP:conf/agere/TaboneF21}, which were omitted from the \fullref{sec:preliminaries}.

\subsection{Term Typing}
\label{sec:appendix-term-typing-rules}

In \Cref{sec:behavioural-typing}, we explained a few term typing rules, including \prooflabel{tBranch} and \prooflabel{tChoice}.
In \Cref{fig:behavioural-typing}, we present the full list of term typing rules.

\begin{figure}[H] 
	\begin{flushleft}
		\framebox[10px+\width]{$ \typingsigma{\Delta}{\Gamma}{\Sigma}{\identifier}{S}{t}{T}{S'} $}
	\end{flushleft}

	\vspace{-2\baselineskip}

	\begin{prooftree}
		\AxiomC{$ \vphantom{\typing{\Delta}{\left( \Gamma, \Gamma_i' \right)}{\identifier}{S}{t_i}{T}{S'}} \typingexpression{\Gamma}{e}{T} $}
		\LeftLabel{\prooflabel{tExpression}}
		\UnaryInfC{$ \typing{\Delta}{\Gamma}{\identifier}{S}{e}{T}{S} $}
	\end{prooftree}

	\vspace{3px}

	\begin{prooftree}
		\hspace*{-1cm}
		\AxiomC{$ \typing{\Delta}{\Gamma}{\identifier}{S}{t_1}{T'}{S'' } $}	
		\AxiomC{$ \typing{\Delta}{\left(\Gamma, x:T' \right)}{\identifier}{S''}{t_2}{T}{S' } $}
		\AxiomC{$ x \neq \identifier $}
		\LeftLabel{\prooflabel{tLet}}
		\TrinaryInfC{$ \typing{\Delta}{\Gamma}{\identifier}{S}{\stexpsequence{x}{t_1}{t_2}}{T}{S' } $}
	\end{prooftree}

	\vspace{3px}

	\begin{prooftree}
		\hspace*{-2cm}
		\AxiomC{$ \forall i \in I \cdot 
		\begin{cases}
			\forall j \in 1..n \cdot \left\{\typingpattern{\identifier}{p_i^j}{T_i^j}{\Gamma_i^j} \right\}  \\
			\typing{\Delta}{\big( \Gamma, \Gamma_i^1 \dotsspace \Gamma_i^n \big)}{\identifier}{S_i}{t_i}{T}{S'}
		\end{cases}
		$}
		\LeftLabel{\prooflabel{tBranch}}
		\UnaryInfC{$ \typing{\Delta}{\Gamma}{\identifier}{\stsessionbranch{? \textlabel{l}_i \big(\widetilde{T_i}\big) . S_i}{ i \in I}}{\stexpreceive{\left\{ \atom{l}_i, \widetilde{p_i} \right\}}{t_i}{i \in I }}{T}{S'} $}
	\end{prooftree}

	\vspace{3px}

	\mycomment{
		But $ w_j $ cannot be of type tuple
	}
	\begin{prooftree}
		\hspace*{-2cm}
		\AxiomC{$ \exists i \in I $}	
		\AxiomC{$ \texttt{l} = \texttt{l}_i $}
		\AxiomC{$ \forall j \in 1..n \cdot \left\{ \typingexpression{\Gamma}{e_j}{T_i^j} \right\} $}
		\LeftLabel{\prooflabel{tChoice}}
		\TrinaryInfC{$ \typing{\Delta}{\Gamma}{\identifier}{\stsessionchoice{! \textlabel{l}_i \big(\widetilde{T_i}\big) . S_i}{i \in I}}{\stexpsend{{\identifier}}{\left\{\atom{l}, e_1 \dotsspace e_n\right\}}}{ \left\{ \texttype{atom}, T_i^1 \dotsspace T_i^n\right\}}{S_i} $}
	\end{prooftree}

%
	\vspace{3px}

	\mycomment{
		No need to check the type of $\identifier$ ($ \typingexpression{\Gamma}{\identifier}{\texttype{pid}}$) because we compare the whole value of $\identifier$ to the typing judgement
	}
	\begin{prooftree}
		\AxiomC{$ \Delta \left(\fun{f}{n}\right) = S \qquad \forall i \in 2..n \cdot \left\{ \typingexpression{\Gamma}{e_i}{T_i} \right\}$}
		\noLine
		\UnaryInfC{$ \Sigma \left(\fun{f}{n}\right) = \Omega \qquad \Omega.\texttt{return\_type} = T \qquad \Omega.\texttt{param\_types} = \widetilde{T} $}
		\LeftLabel{\prooflabel{tRecKnownCall}}
		\UnaryInfC{$ \typing{\Delta}{\Gamma}{\identifier}{S}{\stexpfunction{f}{\identifier, \; e_2 \dotsspace e_{n}}}{T}{\stend} $}
	\end{prooftree}

	\vspace{3px}

	\begin{prooftree}
		\hspace*{-8mm}
		\AxiomC{$ \Sigma\left(\fun{f}{n}\right) = \Omega \qquad \qquad \fun{f}{n} \notin \dom{\Delta} \qquad \qquad \Omega.\texttt{dual} = \dualpid $}
		\noLine
		\UnaryInfC{$ \Omega.\texttt{params} = \widetilde{x} \quad\; \Omega.\texttt{param\_type} =  \widetilde{T} \quad\; \Omega.\texttt{body} = t \quad\; \Omega.\texttt{return\_type} = T  $}
		\noLine
		\UnaryInfC{$  $}	
		\noLine
		\UnaryInfC{$ \typing{\left(\Delta, \fun{f}{n} : S \right)}{\big(\Gamma, \dualpid : \texttype{pid}, \widetilde{x} : \widetilde{T}\big)}{\dualpid}{S}{t}{T}{S'} \qquad \;\; \forall  i \in 2..n \cdot \left\{ \typingexpression{\Gamma}{e_i}{T_i} \right\} $}
		\LeftLabel{\prooflabel{tRecUnknownCall}}
		\UnaryInfC{$ \typing{\Delta}{\Gamma}{\identifier}{S}{\stexpfunction{f}{\identifier, \; e_2 \dotsspace e_{n}}}{T}{S'} $}	
	\end{prooftree}

	\vspace{3px}

	\begin{prooftree}
		\AxiomC{$ \typingexpression{\Gamma}{e}{U} $}
		\noLine
		\UnaryInfC{$ \forall i \in I \qquad \typingpattern{\identifier}{p_i}{U}{\Gamma_i'} \qquad \typing{\Delta}{\left( \Gamma, \Gamma_i' \right)}{\identifier}{S}{t_i}{T}{S'} $}
		\LeftLabel{\prooflabel{tCase}}
		\UnaryInfC{$ \typing{\Delta}{\Gamma}{\identifier}{S}{\stexpcase{e}{p_i}{t_i}{i \in I}}{T}{S'} $}
	\end{prooftree}

	\caption{Term typing rules}
	\label{fig:behavioural-typing}
\end{figure}

\subsection{Expression Typing}
\label{sec:appendix-expression-typing-rules}

Expression are typechecked using the $\typingexpression{\Gamma}{e}{T}$ judgement, which states that ``an expression $e$ has type $T$, subject to the \emph{variable binding} environment $\Gamma$.''
The expression typing rules are listed in \Cref{fig:expression-typing}.

\begin{figure}[H]
		
	\begin{flushleft}
		\framebox[10px+\width]{$ \typingexpression{\Gamma}{e}{T} $}
	\end{flushleft}

	\vspace{-2.8\baselineskip}
	
	\begin{prooftree}
		\AxiomC{$ \forall i \in 1..n \qquad \typingexpression{\Gamma}{e_i}{T_i} $}
		\LeftLabel{\prooflabel{tTuple}}
		\UnaryInfC{$ \typingexpression{\Gamma}{\left\{e_1 \dotsspace e_n\right\}}{\left\{ T_1, \dotsspace T_n \right\}} $}
	\end{prooftree}
	
	\vspace{1px}
	
	\begin{prooftree}
		\AxiomC{$ \typeof{b} \; = \; T $}
		\AxiomC{$ b \neq \listtype{} $}
		\LeftLabel{\prooflabel{tLiteral}}
		\BinaryInfC{$ \typingexpression{\Gamma}{b}{T} $}
		\DisplayProof
		\hspace*{9mm}
		\AxiomC{$ \Gamma \left(x\right) = T $}
		\LeftLabel{\prooflabel{tVariable}}
		\UnaryInfC{$ \typingexpression{\Gamma}{x}{T} $}
	\end{prooftree}
	
	\vspace{1px}

	\begin{prooftree}
		\AxiomC{$ \typingexpression{\Gamma}{e_1}{T} $}
		\AxiomC{$ \typingexpression{\Gamma}{e_2}{\listtype{T}} $}
		\LeftLabel{\prooflabel{tList}}
		\BinaryInfC{$ \typingexpression{\Gamma}{\listtypee{e_1}{e_2}}{\listtype{T}} $}
		\DisplayProof
		\hspace*{9mm}
		\AxiomC{$ \vphantom{\typingexpression{\Gamma}{e_2}{\listtype{T}}} $}
		\LeftLabel{\prooflabel{tEList}}
		\UnaryInfC{$ \typingexpression{\Gamma}{\listtype{}}{\listtype{T}} $}
	\end{prooftree}

	\vspace{1px}

	\begin{prooftree}
		\AxiomC{$ \typingexpression{\Gamma}{e_1}{\texttype{number}} \qquad \typingexpression{\Gamma}{e_2}{\texttype{number}} \qquad \diamond \in \left\{ + , \; - , \; * , \; / \right\}$}
		\LeftLabel{\prooflabel{tArithmetic}}
		\UnaryInfC{$ \typingexpression{\Gamma}{\stexpbinaryop{e_1}{\diamond}{e_2}}{\texttype{number}} $}
	\end{prooftree}
	
	\vspace{1px}

	\begin{prooftree}
		\AxiomC{$ \typingexpression{\Gamma}{e_1}{\texttype{boolean}} \qquad \typingexpression{\Gamma}{e_2}{\texttype{boolean}} \qquad \diamond \in \left\{ \stexpandonly ,\; \stexporonly \right\} $}
		\LeftLabel{\prooflabel{tBoolean}}
		\UnaryInfC{$ \typingexpression{\Gamma}{\stexpbinaryop{e_1}{\diamond}{e_2}}{\texttype{boolean}} $}
	\end{prooftree}

	\vspace{3px}


	\begin{minipage}{1.2\textwidth}	
		\hspace*{6mm}
		\AxiomC{$ \diamond \in \left\{ < , \; > , \; <= , \; >= , \; == , \; ! \! = \right\} $}
		\noLine
		\UnaryInfC{$ \typingexpression{\Gamma}{e_1}{T}  \qquad \typingexpression{\Gamma}{e_2}{T} $}
		\LeftLabel{\prooflabel{tComparisons}}
		\UnaryInfC{$ \typingexpression{\Gamma}{\stexpbinaryop{e_1}{\diamond}{e_2}}{\texttype{boolean}}  $}
		\DisplayProof
		\hspace*{8mm}
		\AxiomC{$ \vphantom{\diamond \in \left\{ < , \; > , \; <= , \; >= , \; == , \; ! \! = \right\}} $}
		\noLine
		\UnaryInfC{$ \typingexpression{\Gamma}{e}{\texttype{boolean}} $}
		\LeftLabel{\prooflabel{tNot}}
		\UnaryInfC{$ \typingexpression{\Gamma}{\stexpnot{e}}{\texttype{boolean}} $}
		\DisplayProof
	\end{minipage}

	\caption{Expression typing rules}
	\label{fig:expression-typing}
\end{figure}

\subsection{Pattern Typing}
\label{sec:appendix-pattern-typing-rules}

New variables may be created using patterns in the \prooflabel{tBranch} and \prooflabel{tCase} rules.
These variables are matched to a type using the judgement, $ \typingpattern{\identifier}{p}{T}{\Gamma}$.
This judgement states that ``a pattern $p$ is matched to type $T$, where it produces new variables and their types are collected $\Gamma$; under the assumption that the variable containing the dual \pid, $\identifier$, remains unchanged.''
The pattern typing rules are found in \Cref{fig:pattern-typing}.

\begin{figure}[H] 
	\begin{flushleft}
		\framebox[10px+\width]{$ \typingpattern{\identifier}{p}{T}{\Gamma}$}
	\end{flushleft}

	\vspace{-0.7\baselineskip}
		
	\begin{minipage}{1.2\textwidth}
		
	\begin{center}
		\hspace*{-25mm}		
		\AxiomC{$ \typingexpression{\emptyset}{b}{T} \qquad b \neq \listtype{}$}
		\LeftLabel{\prooflabel{tpLiteral}}
		\UnaryInfC{$ \typingpattern{\identifier}{b}{T}{\emptyset} $}
		\DisplayProof
		\hspace*{8mm}
		\AxiomC{$ x \neq \identifier$}
		\LeftLabel{\prooflabel{tpVariable}}
		\UnaryInfC{$ \typingpattern{\identifier}{x}{T}{x : T} $}
		\DisplayProof
	\end{center}

	\end{minipage}		

	\vspace{3px}

	\begin{prooftree}
		\AxiomC{$ \forall i \in 1..n $}
		\AxiomC{$ \typingpattern{\identifier}{w_i}{T_i}{\Gamma_i} $}
		\LeftLabel{\prooflabel{tpTuple}}
		\BinaryInfC{$ \typingpattern{\identifier}{\tupletype{w_1 \dotsspace w_n}}{\tupletype{T_1 \dotsspace T_n }}{\Gamma_1 \dotsspace \Gamma_n} $}
	\end{prooftree}
		
	\vspace{3px}

	\begin{minipage}{1.2\textwidth}
		
	\begin{center}
		\hspace*{-30mm}
		\AxiomC{$ \typingpattern{\identifier}{w_1}{T}{\Gamma_1} $}
		\AxiomC{$ \typingpattern{\identifier}{w_2}{\listtype{T}}{\Gamma_2} $}
		\LeftLabel{\prooflabel{tpList}}
		\BinaryInfC{$ \typingpattern{\identifier}{\listtypee{w_1}{w_2}}{\listtype{T}}{\Gamma_1, \Gamma_2} $}
		\DisplayProof
		\hspace*{4mm}
		\AxiomC{$ \vphantom{\typingpattern{\identifier}{w_1}{T}{\Gamma_1}}$}
		\LeftLabel{\prooflabel{tpEList}}
		\UnaryInfC{$ \typingpattern{\identifier}{\listtype{}}{\listtype{T}}{\emptyset} $}
		\DisplayProof
	\end{center}

	\end{minipage}

	\caption{Pattern typing rules}
	\label{fig:pattern-typing}
\end{figure}

\iftoggle{shortversion}{%
}{%
  \section{Proofs}
  \label{sec:appendix-proofs}
  \input{sections/proofs/appendix.tex}
}%

\end{document}